\documentclass[11pt]{article}

\usepackage[top=1in,bottom=1in,left=1in,right=1in]{geometry}
\usepackage[T1]{fontenc}
\usepackage{fullpage}
\usepackage{amsthm}
\usepackage{thmtools}
\usepackage{thm-restate}
\usepackage{footnote}
\usepackage{tabstackengine}

\usepackage{amssymb}
\usepackage{amsmath}
\usepackage{graphicx,color}
\usepackage{amsmath}
\usepackage{mathtools}

\usepackage{enumitem}
\usepackage{subfigure}
\usepackage{epsfig}
\usepackage{minibox}
\usepackage{framed}
\usepackage{algorithm}
\usepackage{algpseudocode}

\usepackage{xcolor}

\usepackage{float}

\allowdisplaybreaks

\newtheorem{theorem}{Theorem}
\newtheorem{lemma}[theorem]{Lemma}
\newtheorem{corollary}[theorem]{Corollary}
\newtheorem{definition}[theorem]{Definition}
\newtheorem{claim}[theorem]{Claim}

\newtheorem{property}{Property}
\newtheorem{criterion}{Criterion}
\newtheorem{invariant}{Invariant}

\algnewcommand\comment[1]{\hfill$\triangleright$ #1}

\def\case #1{\medskip{\bf Case}\ {\it #1:}}

\def\subcase #1{\medskip{\bf Subcase} #1: }
\def\yclaim #1 {\medskip{\bf Claim #1: }}
\long\def\xclaim #1 {\medskip\noindent{\bf Claim } {\it #1}\medskip}
\def\claim #1 #2 {\medskip\noindent{\bf Claim #1.} {\em #2}\medskip}

\newcommand{\R}{\mathbb{R}}

\newcommand{\defi}{\operatorname{def}}
\newcommand{\surp}{\operatorname{surp}}

\newcommand{\Gelig}{G_{elig}}

\newcommand{\Chvatal}{Chv\'{a}tal}

\newcommand{\scale}{\operatorname{scale}}

\begin{document}

\title{Approximate Generalized Matching:\\
$f$-Matchings and $f$-Edge Covers}

\author{Dawei Huang\\
University of Michigan
\and
Seth Pettie\thanks{Supported by
NSF grants CCF-1217338, CNS-1318294, CCF-1514383, CCF-1637546, and CCF-1815316.}\\
University of Michigan}
\date{}

\maketitle
\thispagestyle{empty}
\setcounter{page}{0}

\begin{abstract}
In this paper we present 
almost linear time approximation schemes for several
generalized matching problems on nonbipartite graphs.
Our results include $O_\epsilon(m\alpha(m, n))$-time algorithms
for $(1-\epsilon)$-maximum weight $f$-matching and
$(1+\epsilon)$-approximate minimum weight $f$-edge cover.
As a byproduct, we also obtain direct algorithms for the exact cardinality versions
of these problems running in $O(m\alpha(m, n)\sqrt{f(V)})$ time, where $f(V)$ is the sum of degree constraints 
on the entire vertex set.

The technical contributions of this work include an efficient method for
maintaining {\em relaxed complementary slackness} in generalized matching problems
and approximation-preserving reductions between the $f$-matching and $f$-edge cover problems.
\end{abstract}

\newpage

\section{Introduction}

Many combinatorial optimization problems are known to be reducible to computing optimal matchings
in non-bipartite graphs~\cite{Edm65,Edm65b}.
These problems include computing $b$-matchings, $f$-factors, $f$-edge covers, $T$-joins,
undirected shortest paths (with no negative cycles), and bidirected flows, see \cite{Law01,Sch02,EJ03}.
These problems have been investigated heavily since Tutte's work in the 1950s \cite{Tutte,Pulleyblank}.
However, the existing reductions to graph matching are often inadequate:
they blow up the size of the input~\cite{Law01},
use auxiliary space~\cite{Gab83}, or piggyback on specific matching algorithms~\cite{Gab83}
like the Micali-Vazirani algorithm~\cite{MV80,Var94,Var12}. Moreover, most existing reductions destroy the dual structure of
optimal solutions and are therefore not {\em approximation preserving}.

In this paper we design algorithms for computing $f$-matchings and $f$-edge covers (both defined below)
in a direct fashion, or through efficient, approximation-preserving reductions.  
Because our algorithms are
based on the LP formulations of these problems
(in contrast to approaches using {\em shortest} augmenting walks~\cite{Gab83,MV80,Var94,Var12}),
they easily adapt to {\em weighted} and {\em approximate} variants of the problems.
Let us define these problems formally. Let $G = (V, E)$ be a graph that possibly contains parallel edges and self loops. For any subset $F$ of edges, we use $\deg_F(v)$ to indicate the degree of vertex $v$ in the subgraph induced by $F$. Notice that each self-loop on $v$ that is in $F$ contributes $2$ to this degree. We use $n$ to denote the number of vertices and $m$ to denote the number of edges, counting multiplicities. 
We define $f$-matching and $f$-edge covers as follows.

\begin{description}
\item[$f$-matching] An $f$-matching is a subset $F\subseteq E$ such that $\deg_F(v) \le f(v)$. $F$ is
{\em perfect} if the degree constraints hold with equality. In this case it is also called an \emph{$f$-factor}. 

\item[$f$-edge cover] An $f$-edge cover is a subset $F\subseteq E$ such that $\deg_F(v) \ge f(v)$. It is \emph{perfect} if all degree constraints hold with equality.
\end{description}

The \emph{maximum weight $f$-matching problem} is, given a graph 
$G = (V, E)$ and a weight function $w$ on $E$, to find an $f$-matching $F$ that maximizes $\sum_{e \in F} w(e)$. Similarly, the \emph{minimum weight $f$-edge cover problem} and the \emph{minimum weight $f$-factor problem} ask for an $f$-edge cover and an $f$-factor that minimize their respective weight.

For these three problems, we can assume, without lost of generality, that all the weights are nonnegative for different reasons. For maximum weight $f$-matching, it is safe to ignore any negative weight edges as discarding negative weight edges from $F$ can only improve the solution. For minimum weight $f$-edge cover, any optimum solution must include the set of all negative weight edges. Hence, we can include them into the solution and update the degree constraint accordingly. For minimum weight $f$-factor, since all $f$-factors are of the same size, 
we can translate all the weights by 
$-\min_{e \in E} w(e)$ without changing the optimal solution.

\paragraph{Classic Reductions.}
The classical reduction from $f$-matching to standard graph matching uses the
\emph{$b$-matching} problem as a stepping stone.
A $b$-matching is a function $x : E\rightarrow \mathbb{Z}_{\geq 0}$
(where $x(e)$ indicates \emph{how many} copies of $e$ are in the matching)
such that $\sum_{e \in \delta(v)} x(e) \le b(v)$,
i.e., the number of matched edges incident to $v$, counting multiplicity, is at most $b(v)$.
The maximum weight $f$-matching problem on $G=(V,E,w)$ can be reduced to $b$-matching by subdividing each edge $e=(u,v)\in E$
into a path $(u,u_e,v_e,v)$. Here $u_e,v_e$ are new vertices. We set the weight of the new edges to be $w(u,u_e)=w(v_e,v)=w(u,v) + W$ and $w(u_e,v_e)= 2 W$, where $W$ is the maximum weight in the original graph. The capacity function $b$ is given by $b(u_e)=b(v_e)=1$ for the new vertices and $b(u) = f(u)$ for the original vertices.

To see this reduction correctly reduces the maximum weight $f$-matching problem to the maximum weight $b$-matching problem, we first notice that if the original graph has an $f$-matching $M$ of weight $W^*$, then the new graph must contain a $b$-matching of weight at least $2W^* + 2W m$, where $m$ is the number of edges in the original graph: for every unmatched edge $(u, v)$, we take the edge $(u_e, v_e)$ into the $b$-matching and leave the two edges $(u, u_e)$ and $(v, v_e)$ unmatched. Otherwise we take only the edges $(u, u_e)$ and $(v, v_e)$.  For the other direction, every maximum
weight $b$-matching must be of this form, i.e., in every subdivided
edge, either the middle edge or two outer edges are in the $b$-matching.

This reduction blows up the number of vertices
to $O(m)$ and is not \emph{approximation preserving}.
The $b$-matching problem is easily reduced to standard matching by replicating each vertex $u$ $b(u)$
times, and replacing each edge $(u,v)$ with a bipartite $b(u)\times b(v)$ clique on its endpoints' replicas.
This step of the reduction is approximation preserving, but blows up the number of vertices and edges.
Both reductions together reduce $f$-matching to a graph matching problem on $O(m)$ vertices and
$O(f_{\max} m)$ edges.  Gabow~\cite{Gab83} gave a method for solving $f$-matching in $O(m\sqrt{f(V)})$ time using black-box calls to single iterations of the Micali-Vazirani~\cite{MV80,Var94,Var12} algorithm.

Observe that $f$-matching and $f$-edge cover are {\em complementary} problems:
if $C$ is an $f_C$-edge cover, the complementary edge set $F = E\backslash C$
is necessarily an $f_F$-matching, where $f_F(v) = \deg(v)-f_C(v)$.
Complementarity implies that any polynomial-time algorithm for one problem solves
the other in polynomial time, but it says nothing about the precise complexity of solving
them exactly or approximately.  Indeed, this phenomenon is very well known in the
realm of NP-complete problems.  For example, Maximum Independent Set and
Minimum Vertex Cover are complementary problems, but have completely different approximation profiles: Minimum Vertex Cover has a well-known polynomial time $2$-approximation algorithm, while it is NP-hard to approximate Maximum Independent Set within $n^{1-\epsilon}$ for any $\epsilon > 0$~\cite{Has99, zuckerman07}.
Gabow's $O(m\sqrt{f_F(V)})$ cardinality $f_F$-matching algorithm~\cite{Gab83}
implies that $f_C$-edge cover is computed in $O(m\sqrt{2m - f_C(V)}) = O(m^{3/2})$ time,
and says nothing about the approximability of $f_C$-edge cover.  As far as we are aware,
the fastest approximation algorithms for $f_C$-edge cover (see~\cite{Pothen-etal}) treat it as a
general weighted Set Cover problem on 2-element sets.
\Chvatal's analysis \cite{Chvatal79} shows the greedy algorithm is an $H(2)$-approximation,
where $H(2)=3/2$ is the 2nd harmonic number.

Our interest in the {\em approximate} $f$-edge cover problem is inspired by a new application
to anonymizing data in environments where users have different privacy demands; see~\cite{Pothen-etal,Choromanski-etal,KP16}.  Here the data records correspond to edges and the privacy demand of $v$ is measured by $f(v)$; the goal is to anonymize as few records to satisfy everyone's privacy demands.

\paragraph{New Results.}
We give new algorithms for computing $f$-matchings and $f$-edge covers
approximately and exactly.
\begin{itemize}
\item We give an $O_{\epsilon}(m\alpha(m, n)$-time $(1-\epsilon)$-approximation algorithm for maximum weight $f$-matching problem. The algorithm generalizes the $(1-\epsilon)$-approximate maximum weight matching algorithm by Duan and Pettie~\cite{DP14}
and improves on the $O(f(V)(m + n \log n))$
running time of Gabow~\cite{Gab18}.
The main technical contribution is the application of \emph{relaxed complementary slackness}~\cite{GT89,GT91,DP14} on $f$-matchings, 
and a new DFS-based search procedure for looking for a \emph{maximal} set of \emph{edge-disjoint} augmenting paths in $O(m \alpha(m, n))$ time.

\item We show that a folklore reduction from minimum weight 1-edge cover to maximum weight 1-matching (matching)
is approximation-preserving, in the sense that any $(1-\epsilon)$-approximation for matching gives a $(1+\epsilon)$-approximation
for edge cover.  This implies that 1-edge cover can be $(1+\epsilon)$-approximated in $O_\epsilon(m)$ time~\cite{DP14},
and that one can apply any number of simple and practical algorithms~\cite{DrakeHougardy,Pettie-Sanders,DP14}
to approximate 1-edge cover.  This simple reduction does \emph{not} extend to $f$-matchings/$f$-edge covers when
$f$ is arbitrary.

\item We give an $O_\epsilon(m\alpha(m, n))$-time $(1+\epsilon)$-approximation algorithm for weighted
$f_C$-edge cover, for any $f_C$.
Our algorithm follows from two results, both of which are somewhat surprising.  First, any approximate weighted
$f_F$-matching algorithm {\em that reports a $(1\pm \epsilon)$-optimal dual solution} can be transformed into a $(1+O(\epsilon))$-approximate weighted $f_C$-edge cover algorithm.
Second, such an $f_F$-matching algorithm exists, and its running time is $O_\epsilon(m\alpha(m, n))$.
The first claim is clearly false if we drop the approximate dual solution requirement (for the same reason that an $O(1)$-approximate vertex cover does not translate into an $O(1)$-approximate maximum independent set),
and the second is surprising because the running time is independent of the demand function $f_F$ and the magnitude of the edge weights.

\item As corollaries of these reductions, we obtain a new exact algorithm for minimum cardinality $f_C$-edge cover
running in $O(m\alpha(m, n)\sqrt{f_C(V)})$ time, rather than $O(m^{3/2})$ time (\cite{Gab83}),
and a direct algorithm for cardinality $f_F$-matching that runs in $O(m\alpha(m, n)\sqrt{f_F(V)})$ time,
without reduction~\cite{Gab83} to the Micali-Vazirani algorithm~\cite{MV80,Var94,Var12}.
\end{itemize}

The blossom structure and LP characterization of $b$-matching is considerably simpler
than the corresponding blossoms/LPs for $f$-matching and $f$-edge cover.
In the interest of simplicity,
one might want efficient code that solves (approximate) $b$-matching directly, without
viewing it as a special case of the $f$-matching problem.\footnote{The $b$-matching problem can be regarded as an $f$-matching problem on a multigraph in which there is implicitly an infinite supply of each edge.}
We do not know of such a direct algorithm.  Indeed, the structure of $b$-matching blossoms seems
to rely on strict complementary slackness, and is \emph{incompatible} with our main technical
tool, relaxed complementary slackness.\footnote{Using relaxed complementary slackness, matched and unmatched
edges have \emph{different} eligibility criteria (to be included in augmenting paths and blossoms) whereas $b$-matching blossoms require that all copies of an edge---matched and unmatched alike---are all eligible or all ineligible.} Thus, for somewhat technical reasons,
we are forced to solve approximation $b$-matching
using more sophisticated $f$-matching tools.

\paragraph{Comparison to Previous Results.}
Our linear time $(1-\epsilon)$-approximation algorithm for maximum weight $f$-matching can be seen as a direct generalization of the Duan-Pettie algorithm for approximate maximum weight matching ($1$-matching)~\cite{DP14}. The key technical ingredient  is the generalization of \emph{relaxed complementary slackness}, see \cite{GT89,GT91,DP14}, to $f$-matching, and a corresponding implementation of Edmonds' Search with relaxed complementary slackness. 
The former relies heavily on the ideas (blossoms, augmenting walks) defined by Gabow~\cite{Gab18}. 
Our implementation of
Edmonds' search involves finding augmenting walks in batches. The procedure of \cite[\S 8]{GT91} for matching finds a maximal set of vertex-disjoint augmenting paths.
We develop a corresponding procedure that finds a maximal set of edge-disjoint augmenting walks \underline{\emph{and}} cycles. 
Including alternating cycles in the output allows us to conduct the search in near linear time, and keep the search more organized and tree-structured.\footnote{These issues only arise when finding augmenting paths in \emph{batches}, 
not one-at-a-time~\cite{Gab18},
and only when the problem is $f$-matching,
not matching~\cite{DuanPS14,GT91}.}

\paragraph{Structure of the Paper.}
In Section~\ref{sec:basis} we give an introduction to the LP-formulation of generalized matching
problems and Gabow's formulation \cite{Gab18} for their blossoms and augmenting walks.
In Section~\ref{sec:reduction-1} we show that a folklore reduction from 1-edge cover to 1-matching is approximation-preserving
and in Section~\ref{sec:reduction-2} we reduce approximate $f$-edge cover to approximate $f$-matching. In Section~\ref{sec:approx} we
give an $O(Wm\alpha(m, n)\epsilon^{-1})$-time algorithm for $(1-\epsilon)$-approximate $f$-matching in graphs with weights in $[0,W]$ and then speed it up to $O(m\alpha(m, n) \epsilon^{-1} \log{\epsilon^{-1})}$, independent of the weight function. Section~\ref{sec:aug-walk} gives $O(m\alpha(m, n))$ algorithm to compute a maximal set of augmenting walks and alternating cycles; cf.~\cite[\S 8]{GT91}.

\section{Basis of $f$-matching and $f$-edge cover} \label{sec:basis}

This section reviews basic algorithmic concepts from matching theory and their generalizations to the $f$-matching and $f$-edge cover problems, e.g., LPs, blossoms, and augmenting walks.
These ideas lay the foundation for generalizing the Duan-Pettie algorithm \cite{DP14} for Approximate Maximum Weight Matching
to Approximate Maximum Weight $f$-Matching and Approximate Minimum Weight $f$-Edge Cover.

\paragraph{Notation.}
The input is a multigraph $G = (V, E)$ with a \emph{nonnegative} weight function $w: E \mapsto \mathbb{R}_{\geq 0}$. For any vertex $v$, define $\delta(v)$ and $\delta_0(v)$ be the set of non-loop edges and self-loops, respectively, incident on $v$. For $S \subseteq V$, let $\delta(S)$ and $\gamma(S)$ be the sets of edges with exactly one endpoint and both endpoints in $S$, respectively, so
$\delta_0(v) \subseteq \gamma(S)$ if $v\in S$.
For $T \subseteq E$, $\delta_T(S)$ denotes the intersection of $\delta(S)$ and $T$. By definition, $\deg_T(S) = |\delta_T(S)|$.

\subsection{LP formulation}

The maximum weight $f$-matching problem can be expressed as maximizing $\sum_{e \in E}  w(e) x(e)$, subject to the following constraints:
\begin{equation} \label{LP:F}
  \begin{aligned}
    &\sum_{e \in \delta(v)} x(e) + \sum_{e \in \delta_0(v)} 2x(e)\leq f(v), \mbox{ for all $v \in V$, } \\
    &\sum_{e \in \gamma(B) \cup I} x(e) \leq \left\lfloor \frac{f(B) + |I|}{2} \right\rfloor, \mbox{ for all $B \subseteq V, I \subseteq \delta(B)$,}\\
    &0 \leq x(e) \leq 1, \mbox{ for all $e \in E$.}
  \end{aligned}
\end{equation}

Here, the blossom constraint $\sum_{e \in \gamma(B) \cup I} x(e) \leq \left\lfloor \frac{f(B) + |I|}{2} \right\rfloor$ is a generalization of blossom constraint $\sum_{e \in \gamma(B)} x(e) \leq \left\lfloor \frac{|B|}{2} \right\rfloor$ in ordinary matching. The reason that we have a subset $I$ of incident edges in the sum is that the subset allows us to distinguish between matched edges that have both endpoints inside $B$ with those with exactly one endpoint. Any basic feasible solution $x$ of this LP is integral \cite[\S 33]{Sch02}, and can therefore be interpreted as a membership vector of an $f$-matching $F$. To certify (approximate) optimality of a solution, the algorithm works with the dual LP, which is:

\begin{equation} \label{LP:F-dual}
  \begin{aligned}
    \mbox{minimize}   \;\; & \sum_{v \in V} f(v)y(v) + \sum_{B \subseteq V, I \subseteq \delta(B)} \left\lfloor \frac{f(B) + |I|}{2} \right\rfloor z(B, I) + \sum_{e} u(e),\\
    \mbox{subject to} \;\; & yz_F(e) + u(e) \geq w(e), \mbox{ for all $e \in E$,} \\
                           & y(v) \geq 0, z(B, I) \geq 0, u(e) \geq 0.
  \end{aligned}
\end{equation}

Here the aggregated dual $yz_F: E \mapsto \R_{\geq 0} $ is defined as:
\begin{align*}
  yz_F(u, v) = y(u) + y(v) + \sum_{\substack{B, I: (u, v) \in \gamma(B) \cup I,\\ I \subseteq \delta(B))}} z(B, I).
\end{align*}

Notice that $u$ can be equal to $v$ when the edge is a self-loop.  Unlike matching, each $z$-value here is associated with the combination of a vertex set $B$ and a subset $I$ of its incident edges.

The minimum weight $f$-edge cover problem can be expressed as minimizing $\sum_{e \in E} w(e) x(e)$, subject to:
\begin{equation} \label{LP:EC}
  \begin{aligned}
     &\sum_{e \in \delta(v)} x(e) + \sum_{e \in \delta_0(v)} 2x(e) \geq f(v), \mbox{ for all $v \in V$,} \\
     &\sum_{e \in \gamma(B) \cup (\delta(B) \setminus I)} x(e) \geq \left\lceil \frac{f(B) - |I|}{2} \right\rceil, \mbox{ for all $B \subseteq V$ and $I \subseteq \delta(B)$,} \\
     &0 \leq x(e) \leq 1, \mbox{for all $e \in E$.}
   \end{aligned}
\end{equation}

With the dual program being:

\begin{equation} \label{LP:EC-dual}
  \begin{aligned}
    \mbox{maximize}   \;\; & \sum_{v \in V} f(v)y(v) + \sum_{B \subseteq V, I \subseteq \delta(B)} \left\lceil \frac{f(B) - |I|}{2} \right\rceil z(B, I) - \sum_{e \in E} u(e),\\
    \mbox{subject to} \;\; & yz_C(e) - u(e) \leq w(e), \mbox{ for all $e \in E$,} \\
                           & y(v) \geq 0, z(B, I) \geq 0, u(e) \geq 0,
  \end{aligned}
\end{equation}

where
\begin{equation*}
yz_C(u, v) = y(u) + y(v) + \sum_{\substack{B, I: (u, v) \in \gamma(B) \cup (\delta(B) \setminus I)\\ I \subseteq \delta(B)}} z(B, I). \footnote{We use $yz_F$ and $yz_C$ to denote the aggregated dual $yz$ for $f$-matching and $f$-edge cover respectively. We will omit the subscript if it is clear from the context.}
\end{equation*}

Following Gabow~\cite{Gab18},
both of our $f$-matching and $f$-edge cover algorithms maintain a dynamic \emph{feasible} solution $F \subseteq E$ that satisfies the primal constraints. We call edges in $F$ \emph{matched} and all other edges \emph{unmatched}, which is referred to as the \emph{type} of an edge.
A vertex $v$ is \emph{saturated} if $\deg_F(v) = f(v)$. It is \emph{unsaturated/oversaturated} if $\deg_F(v)$ is smaller/greater than $f(v)$. Given an $f$-matching $F$, the \emph{deficiency} $\defi(v)$ of a vertex $v$ is defined as $\defi(v) = f(v) - \deg_F(v)$. Similarly, for an $f$-edge cover $C$, the \emph{surplus} of a vertex is defined as $\surp(v) = \deg_C(v) - f(v)$.

\subsection{Blossoms}

We follow Gabow's \cite{Gab18} definitions and terminology for $f$-matching blossoms, augmenting walks, etc. A \emph{blossom} is a tuple $(B, E_B, \beta(B), \eta(B))$ where $B$ is the vertex set, $E_B$ is the edge set, $\beta(B) \in B$ is the \emph{base vertex}, and $\eta(B) \subset \delta(\beta(B)) \cap \delta(B)$, $|\eta(B)| \leq 1$, is the \emph{base edge set}, which may be empty. We often refer to the blossom by referring to its vertex set $B$. Blossoms can be defined inductively as follows.

\begin{restatable}{definition}{blossoms}\cite[Definition~4.2]{Gab18} \label{def:blossom}
A single vertex $v$ forms a \emph{trivial blossom}, or a \emph{singleton}. Here $B = \{v\}$, $E_B = \emptyset$, $\beta(B) = v$, and $\eta(B) = \emptyset$.

Inductively, let $B_0, B_1, \ldots, B_{l-1}$ be a sequence of disjoint singletons or nontrivial blossoms. Suppose there exists a closed walk $C_B = \{ e_0, e_1, \ldots, e_{l-1}\} \subseteq E$ starting and ending with $B_0$ such that $e_i \in B_i \times B_{i+1 \pmod{l}} $. The vertex set $B = \bigcup_{i = 0}^{l - 1} B_i$ is identified with a blossom if the following are satisfied:
\begin{enumerate}[itemsep=0ex]
  \item {\it Base Requirement:} If $B_0$ is a singleton, the two edges incident to $B_0$ on $C_B$, i.e., $e_0$ and $e_{l-1}$, must both be matched or both be unmatched.
  \item {\it Alternation Requirement:} Fix a $B_i, i \neq 0$. If $B_i$ is a singleton, exactly one of $e_{i-1}$ and $e_i$ is matched. If $B_i$ is a nontrivial blossom, then
  $\eta(B_i) \neq \emptyset$ and must be either $\{e_{i-1}\}$ or $\{e_i\}$.
\end{enumerate}

The edge set of the blossom $B$ is $E_B = C_B \cup (\bigcup_{i = 0}^{l - 1} E_{B_i})$ and its base is $\beta(B) = \beta(B_0)$. If $B_0$ is not a singleton, $\eta(B) = \eta(B_0)$. If $B_0$ is a singleton, $\eta(B)$ may either be empty or contain one edge, which is in $\delta(B) \cap \delta(B_0)$ that is the opposite type of $e_0$ and $e_{l-1}$.
\end{restatable}

Blossoms are classified as either \emph{light} or \emph{heavy} \cite[p.~32]{Gab18}. If $B_0$ is a singleton, $B$ is light (heavy) if $e_0$ and $e_{l-1}$ are both unmatched (both matched). Otherwise, $B$ is light or heavy iff $B_0$ is light or heavy.  Note that blossoms in the ordinary matching problem ($1$-matching) are always light, 
since no vertex is adjacent to 2 matched edges.

One purpose of blossoms is to identify parts of graph
that can be contracted and treated \emph{similar} to
individual vertices when searching for augmenting walks.
This is formalized by Lemma~\ref{lem:blossom-walk},
which can be seen as a restatement of Lemma 4.4 from \cite{Gab18} for $f$-matchings.

\begin{lemma} \label{lem:blossom-walk}
Let $v$ be an arbitrary vertex in $B$. There exists an even length alternating walk $P_0(v)$ (whose length could be $0$) and an odd length alternating walk $P_1(v)$ from $\beta(B)$ to $v$ using edges in $E_B$. Moreover, the terminal edge
incident to $\beta(B)$, if it exists,
must have a different type than the edge in $\eta(B)$,
if any. In other words, this
edge must be matched if $B$ is heavy and
unmatched if $B$ is light.
\end{lemma}

\begin{proof}
We prove this by induction. The base case is a blossom $B$ consisting of a cycle of singletons $\left<v_0, v_1, \ldots, v_{l-1}\right>$, where $v = v_{i}$ for some $0 \leq i < l$. Then one of the two walks $\left<v_0, v_1, \ldots, v_i\right>$ and $\left<v_0, v_{l-1}, v_{l-2}\ldots, v_{i}\right>$
must be odd and the other must be even.

Now for the inductive step: Consider the cycle $C_B = \langle B_0, e_0, B_1, \ldots, e_{l-2}, B_{l-1} e_{l-1}, B_0\rangle$ where the $\{B_i\}$s are either singletons or contracted blossoms. Suppose the claim holds inductively for all nontrivial blossoms in $B_0, B_1, \ldots, B_{l-1}$. Let $v$ be an arbitrary vertex in $B$. We use $P_{B_i, j}(u)$ ($0 \leq i < l, j \in \{0, 1\}, u \in B_i$) to denote the walk $P_0(u)$ and $P_1(u)$ guaranteed in blossom $B_i$. There are two cases:

\case{1} When $v$ is contained in a singleton $B_k$. We examine the two walks $\widehat{P} = \langle B_0, e_0, B_1, e_1, \ldots, e_{k-1}, B_k \rangle$ and $\widehat{P}' = \langle B_0, e_{l-1}, B_{l-1}, e_{l-2}, \ldots, e_k, B_k \rangle$. Notice that $\widehat{P}$ and $\widehat{P}'$ are walks in the graph obtained by contracting all subblossoms $B_0, B_1, \ldots, B_{l-1}$ of $B$. By the inductive hypothesis, we can extend $\widehat{P}$ and $\widehat{P}'$ to $P$ and $P'$ in the original graph $G$ by replacing each $B_i$ with the walk in the original graph connecting the endpoints of $e_{i-1}$ and $e_{i}$ of the appropriate parity. In particular, if $e_{i-1}$ and $e_{i}$ are of different types, we replace $B_i$ with the even length walk guaranteed by the induction hypothesis. Otherwise, we replace it with the odd length walk. Notice that by the alternation requirement, one of $P$ and $P'$ must be odd and the other must be even.

\case{2} When $v$ is contained in a non-trivial blossom $B_k$, $0 \leq k < l$. Without loss of generality, $e_{k-1} = \eta(B_k)$. Consider the contracted walk $\widehat{P} = \langle e_0, e_1, \ldots, e_{k-1}\rangle$. We extend $\widehat{P}$ to an alternating walk $P$ in $E_B$ terminating at $e_{k-1}$ similar to Case 1. Then $P_{0}(v)$ and $P_1(v)$ are obtained by concatenating $P$ with the alternating walk $P_{B_k, 0}(v)$ or $P_{B_k, 1}(v)$, whichever has the right parity.

Notice that in both cases, the \emph{base requirement} in Definition~\ref{def:blossom} guarantees the starting edge of both alternating walks $P_0(v)$ and $P_1(v)$ alternates with 
the base edge $\eta(B)$.
\end{proof}

The main difference between blossoms in generalized matching problems and blossoms in ordinary matching is that $P_0(v)$ and $P_1(v)$ are \emph{both} meaningful for finding augmenting walks or blossoms. In ordinary matching, since each vertex has at most $1$ matched edge incident to it, an alternating walk enters the blossom at the base vertex via a matched edge and must leave with an unmatched edge.
As a result the subwalk inside the blossom is always even.
In generalized matching problems, this subwalk can be either even or odd, and may contain a cycle.
In general, an alternating walk enters the blossom at the base edge and can leave the blossom at \emph{any} nonbase edge.

Similar to ordinary matching algorithms, we \emph{contract} blossoms in order to find augmenting structures to improve our $f$-matching. Contracting a blossom $B$ means replacing $B$ with a single vertex $v$ with an $f$-value $f(v) = \sum_{v' \in B} f(v') - 2 |M \cap E[B]|$. Here $E[B]$ is the set of edges induced by the vertex set $B$.

Next we extend to notion of \emph{maturity} from \cite[p.~43]{Gab18} to $f$-matching and $f$-edge cover. Let us focus on $f$-matching first. Due to complementary slackness, we can only assign a positive $z$-value for the pair $(B, I)$ if it
satisfies the constraint
$|F \cap (\gamma(B) \cup I)| \leq \left\lfloor (f(B) + |I|)/2 \right\rfloor$ with equality. For ordinary matching, this requirement is implied by the combinatorial definition of blossoms. However, this is not the case for generalized matching, so we need a blossom to be \emph{mature} to fulfill the complementary slackness property.

\begin{definition}[Mature Blossom] \label{def:maturity}
A blossom is \emph{mature} w.r.t an $f$-matching $F$ if it satisfies the following:
\begin{enumerate}[itemsep=0pt]
  \item Every vertex $v \in B \setminus \{\beta(B)\}$ is saturated.
  \item $\defi(\beta(B)) = 0$ or $1$. If  $\defi(\beta(B)) = 1$, $B$ must be a light blossom and $\eta(B) = \emptyset$; If $\defi(\beta(B)) = 0$, $\eta(B) \neq \emptyset$.
\end{enumerate}
\end{definition}

The algorithm only contracts and manipulates mature blossoms. The definition for maturity is motivated by the requirement that a blossom processed by the algorithm must satisfy the following two properties:
\begin{itemize}[itemsep=0pt]
  \item Complementary slackness: A dual variable can be positive only if its primal constraint is satisfied with equality. In our algorithm, a blossom can have a positive $z$-value only if $|F \cap (\gamma(B) \cup I(B))| = \left\lfloor \frac{f(B) + |I(B)|}{2}\right\rfloor$, for a particular $I(B) \subseteq \delta(B)$ that we are going to define momentarily.
  \item Topology of augmenting walks: An augmenting walk in $G$ can only start with an unmatched edge. As a result, an augmenting walk in the contracted graph must start with a singleton or an unsaturated light blossom. If a blossom is unsaturated, it must be eligible to start an augmenting walk, and thus must be light.
\end{itemize}

According to Definition~\ref{def:maturity}, a mature blossom cannot be both heavy and unsaturated. Now we show that a mature blossom satisfies its corresponding primal constraint with equality. To show this fact, we first define the set $I(B)$
associated with blossom $B$~\cite[p.44]{Gab18},
which is the one for which $z(B,I(B))$ is perhaps non-zero.
\[
I(B) = \delta_F(B) \oplus \eta(B).
\]
Here $\oplus$ is the symmetric difference operator (XOR).
All other subsets $I$ of $\delta(B)$ will have $z(B, I) = 0$.
If $B$ is a mature blossom, then we have $|F \cap (\gamma(B) \cup I(B))| = \left\lfloor \frac{f(B) + |I(B)|}{2} \right\rfloor$

\begin{lemma} \label{lem:mature-factor}
If an $f$-matching blossom $B$ is mature, we have $|F \cap (\gamma(B) \cup I(B))| = \left\lfloor \frac{f(B) + |I(B)|}{2} \right\rfloor$.
\end{lemma}
\begin{proof}
We first sketch the idea of the proof.
Assume for simplicity that the deficiency is $0$ for every vertex $v \in B$, i.e.,
there are exactly $f(v)$ matched edges
incident to $v$, and every edge in $I(B)$ is matched. Then every matched edge $e \in F \cap \gamma(B)$ contributes $2$ to $f(B)$, one for each endpoint, and every edge $e \in F\cap I(B)$
contributes $1$ to $f(B)$ and 1 to $|I(B)|$. Thus we have:
\begin{equation*}
    2|F \cap (\gamma(B) \cup I(B))| = f(B) + |I(B)|.
\end{equation*}

Now we eliminate the assumption by a case analysis for the deficiency of $\beta(B)$. If $\defi(\beta(B)) = 0$, the assumption on deficiency holds, while all but possibly $1$ edge in $I(B)$, namely the base edge, are matched. This makes $f(B) + |I(B)|$ at least $2|F \cap (\gamma(B) \cup I(B))|$ and at most $2|F \cap (\gamma(B) \cup I(B))| + 1$, and the equality $|F \cap (\gamma(B) \cup I(B))| = \left\lfloor \frac{f(B) + |I(B)|}{2} \right\rfloor$ follows.

When $\defi(\beta(B)) = 1$, since $\eta(B) = \emptyset$, the assumption that $I(B)$ only contains matched edges holds. Since exactly $1$ of $B$'s vertices has deficiency $1$, we have:
\begin{equation*}
    2|F \cap (\gamma(B) \cup I(B))| + 1 = f(B) + |I(B)|
\end{equation*}
And the equality $|F \cap (\gamma(B) \cup I(B))| = \left\lfloor \frac{f(B) + |I(B)|}{2} \right\rfloor$ follows.
\end{proof}

We complete the discussion by giving the definition for maturity and the corresponding properties for mature blossoms in $f$-edge covers.

\begin{definition}[Mature Blossom for $f$-edge cover] \label{def:ec-maturity}
A blossom is \emph{mature} w.r.t an $f$-edge cover $F$ if it satisfies the following:
\begin{enumerate}[itemsep=0pt]
  \item Every vertex $v \in B \setminus \{\beta(B)\}$ is saturated: $\deg_F(v) = f(v)$.
  \item $\surp(\beta(B)) = 0$ or $1$. If $\surp(\beta(B)) = 1$, $B$ must be a heavy blossom and $\eta(B) = \emptyset$; If $\surp(\beta(B)) = 0$, $\eta(B) \neq \emptyset$.
\end{enumerate}
\end{definition}

\begin{lemma} \label{lem:mature-cover}
If an $f$-edge-cover blosom $B$ is mature, we have $|F \cap (\gamma(B) \cup (\delta(B) \setminus I(B)))| = \left\lceil \frac{f(B) - |I(B)|}{2} \right\rceil$.
\end{lemma}

\begin{figure}
  \centering
  \includegraphics[scale=1.0]{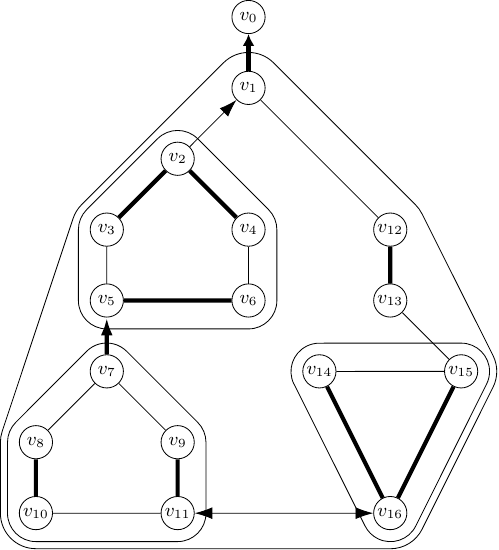}\\
  \vspace{20pt}
  \includegraphics[scale=1.0]{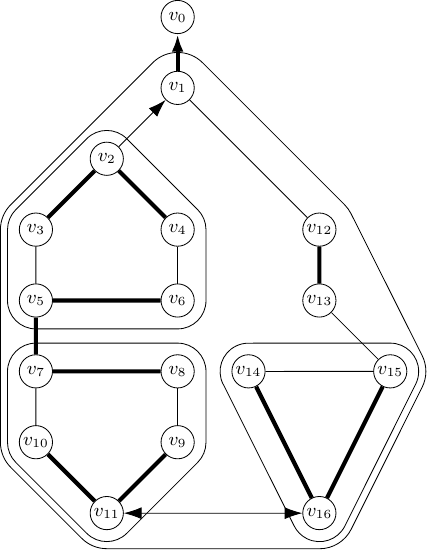}
  \caption{Two examples of contractible blossoms: Bold edges are matched and thin ones are unmatched. Blossoms are circled with a border. Base edges are represented with arrow pointing away from the blossom.}
\end{figure}

\begin{figure}
  \centering
  \includegraphics[scale=1.0]{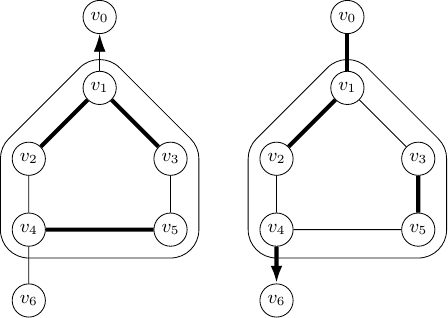}
  \caption{An example for how a blossom changes with an augmentation: here the augmenting walk is $\langle v_0, v_1, v_3, v_5, v_4, v_6\rangle$. Notice that after rematching, the base edge of the blossom changes from $(v_0, v_1)$ to $(v_4, u_6)$, and the blossom turns from a heavy blossom to a light one.} \label{fig:blossom2}
\end{figure}

\subsection{Augmenting/Reducing Walks}

Augmenting walks are analagous to augmenting paths from ordinary matching.
Complications arise from the fact that an $f$-matching blossom cannot be treated identically to a single vertex after it is contracted.
For example, in Figure 2, the two edges $(v_0, v_1)$ and $(v_4, v_6)$ incident to blossom $\left\{v_1, v_2, v_4, v_5, v_3\right\}$ are
of the same type, both before and after augmenting along the walk $\left<v_0, v_1, v_3, v_5, v_4, v_6\right>$.  This can never happen in ordinary matching!  Moreover, augmenting walks can begin and end at the same vertex and can visit the same vertex multiple times. Hence a naive contraction of a blossom into a single vertex loses key information about the internal structure of blossoms.
Definition~\ref{def:AP-f-factor}, taken from Gabow~\cite[p.28, p.44]{Gab18} characterizes when a walk in the contracted graph can be extended to an augmenting walk.

\begin{restatable}{definition}{augmentingwalk} \label{def:AP-f-factor}
Let $\widehat{G}$ be the graph obtained from $G$ by contracting a laminar set $\Omega$ of blossoms. Let $\widehat{P} = \langle B_0, e_0, B_1, e_1, \ldots, B_{l-1}, e_{l-1}, B_l\rangle$ be a walk in $\widehat{G}$. Here $\{e_i\}$ are edges and $\{B_i\}$ are nontrivial blossoms or singletons, with $e_i \in B_i \times B_{i+1}$ for all $0 \leq i < l$. We say $\widehat{P}$ is an \emph{augmenting walk} with respect to the $f$-matching $F$ if the following requirements are satisfied:
\begin{enumerate}[itemsep=0ex]
  \item {\it Terminal Vertices Requirement:} The terminals $B_0$ and $B_l$ must be unsaturated singletons or unsaturated light nontrivial blossoms. If $P$ is a closed walk ($B_0 = B_{l}$), $B_0$ must be a singleton and $\defi(\beta(B_0)) \geq 2$.
  \item {\it Terminal Edges Requirement:} If the terminal vertex $B_0$ ($B_l$) is a singleton, the incident terminal edge $e_0$ ($e_{l-1}$) must be unmatched. Otherwise it can be either matched or unmatched.
  \item {\it Alternation Requirement:} Let $B_i, 0 < i < l$, be an internal blossom. If $B_i$ is a singleton, exactly one of $e_{i-1}$ and $e_i$ is matched. If $B_i$ is a nontrivial blossom, $\eta(B_i) \neq \emptyset$ and must be one of $\{e_{i-1}\}$ or $\{e_{i}\}$.
\end{enumerate}
\end{restatable}

A natural consequence of the above definition is that an augmenting walk $\widehat{P}$ in $\widehat{G}$ can be extended to an augmenting walk $P$ in $G$. This is proved exactly as in Lemma~\ref{lem:blossom-walk}. We call $P$ the \emph{preimage} of $\widehat{P}$ in $G$ and $\widehat{P}$ the image of $P$ in $\widehat{G}$.

\begin{definition}
Let $\widehat{P}$ be an augmenting walk in $\widehat{G}$. An \emph{augmentation} along $\widehat{P}$ makes the following changes to $F$ and $\Omega$.
\begin{enumerate}[itemsep=0ex]
  \item Let $P$ be the preimage of $\widehat{P}$ in $G$. Update $F$ to $F \oplus P$.
  \item If $B \in \Omega$ is a blossom intersecting $P$, we set $\eta(B) \leftarrow (P \cap \delta(B)) \setminus \eta(B)$ and set $\beta(B)$ to the vertex in $B$ that is incident to the edge in $\eta(B)$. Notice that $|P \cap \delta(B)| = 1$ or $2$, and in the case when $|P \cap \delta(B)| = 1$, we must have $\eta(B) = \emptyset$.
\end{enumerate}
\end{definition}

Some remarks can be made here regarding connection to augmenting walks and mature blossoms.
\begin{itemize}
    \item A blossom that is not mature may contain an augmenting walk. Specifically, suppose $B$ is light and unsaturated. If any nonbase vertex $v \neq \beta(B)$ in $B$ is also unsaturated, the odd length alternating walk from $\beta(B)$ to $v$ satisfies the definition of an augmenting walk. Alternatively, if $\beta(B)$ has deficiency of $2$ or more, the odd length alternating walk from $\beta(B)$ to $\beta(B)$ is also augmenting.
  For these reasons, the algorithm is designed such that immature blossoms are never contracted.
    \item Augmentation never destroys maturity. In particular, it never creates an unsaturated heavy blossom. As a result, all blossoms we maintain stay mature throughout the entirety of the algorithm.
\end{itemize}

In $f$-edge cover, the corresponding notion is called \emph{reducing walk}. The definition of reducing walk can be naturally obtained from Definition~\ref{def:AP-f-factor} while replacing ``unsaturated'', ``deficiency'', and  ``light'' with ``oversaturated'', ``surplus'', and ``heavy'', and exchanging ``matched'' and ``unmatched''.
It is also worth pointing out that if an $f$-matching $F$ and an $f'$-edge cover $F'$ are \emph{complement to} each other, i.e., $F' = E \setminus F$ and $f(v) + f'(v) = \deg(v)$, and they have the same blossom set $\Omega$, then an augmenting walk $\widehat{P}$ for $F$ is also a reducing walk for $F'$.

\subsection{Complementary Slackness} \label{sec:CS}

To characterize an (approximately) optimal solution, we maintain dual functions: $y: V \mapsto \R_{\geq 0}$ and $z: 2^V \mapsto \R_{\geq 0}$. Here $z(B)$ is short for $z(B, I(B))$. We do not explicitly maintain the edge dual $u: E \mapsto \R_{\geq 0}$ since its minimizing value can be explicitly given by $u(e) = \max\{w(e) - yz(e), 0\}$. For $f$-matching $F$, the following property characterizes an approximate maximum weight $f$-matching:

\begin{property}[Approximate Complementary Slackness for $f$-matching] ~\label{pro:F-approx}
Let $\delta_1, \delta_2 \geq 0$ be nonnegative parameters. We say an $f$-matching $F$, duals $y, z,$ and the set of blossoms $\Omega$ satisfies $(\delta_1, \delta_2)$-approximate complementary slackness if the following hold:
\begin{enumerate}[itemsep=0ex]
  \item {\it Approximate Domination.} For each unmatched edge $e \in E \setminus F$, $yz(e) \geq w(e) - \delta_1$.
  \item {\it Approximate Tightness.} For each matched edge $e \in F$, $yz(e) \leq w(e) + \delta_2$.
  \item {\it Blossom Maturity.} For each blossom $B \in \Omega$, $|F \cap (\gamma(B) \cup I(B))| = \left\lfloor \frac{f(B) + |I(B)|}{2}\right\rfloor$.
  \item {\it Unsaturated Vertices' Duals.} For each unsaturated vertex $v$, $y(v) = 0$.
\end{enumerate}
\end{property}

\begin{lemma} ~\label{lem:F-approx}
Let $F$ be an $f$-matching in $G$ along with duals $y, z$ and let $F^*$ be the maximum weight $f$-matching. If $F, \Omega, y, z$ satisfy Property~\ref{pro:F-approx} with parameters $\delta_1$ and $\delta_2$, we have
\[
w(F) \geq w(F^*) - \delta_1 |F^*| - \delta_2|F|.
\]
\end{lemma}
\begin{proof}
We first define $u: E \mapsto \R$ as
\begin{equation*}
u(e) =
\begin{cases}
w(e) - yz(e) + \delta_2, &\mbox{ if $e \in F$.} \\
0,            &\mbox{ otherwise.}
\end{cases}
\end{equation*}
From approximate tightness, we have $u(e) \geq 0$ for all $e \in E$. Therefore, $yz(e) + u(e) \geq w(e) - \delta_1$ for \underline{all} $e \in E$ and $yz(e) + u(e) = w(e) + \delta_2$ for all $e \in F$. This gives the following:
\begin{align*}
  w(F) &= \sum_{e \in F} w(e) = \sum_{e \in F} (yz(e) + u(e) - \delta_2) \\
       &= \sum_{v \in V} \deg_F(v) y(v) + \sum_{B \in \Omega} |F \cap (\gamma(B) \cup I(B))| z(B) + \sum_{e \in F} u(e) - |F|\delta_2 \\
\intertext{By Property~\ref{pro:F-approx} (Unsaturated Vertices' Duals, Blossom Maturity, and the definition of $u$), this is equal to }
       &= \sum_{v \in V} f(v) y(v) + \sum_{B \in \Omega} \left\lfloor \frac{f(B) + |I(B)|}{2}\right\rfloor z(B) + \sum_{e \in E} u(e) - |F| \delta_2\\
       &\geq \sum_{v \in V} \deg_{F^*}(v) y(v) + \sum_{B \in \Omega} |F^* \cap (\gamma(B) \cup I(B))| z(B) + \sum_{e \in F^*} u(e) - |F| \delta_2 \\
       &= \sum_{e \in F^*}( yz(e) + u(e) ) - |F| \delta_2 \\
       & \geq \sum_{e \in F^*} (w(e) - \delta_1) - |F| \delta_2 = w(F^*) - |F^*| \delta_1 - |F| \delta_2.
\end{align*}
\end{proof}

We can easily extend the proof of Lemma~\ref{lem:F-approx} to show that if we have multiplicative errors for approximate domination/tightness, $F$ is an approximately optimal solution. Formally, if we have the following multiplicative version of Property~\ref{pro:F-approx}:

\begin{property}[Approximate Complementary Slackness for $f$-matching with Multiplicative Error] ~\label{pro:F-mult-approx}
Let $0 \leq \epsilon_1, \epsilon_2 < 1$ be nonnegative parameters.
We say an $f$-matching $F$, duals $y, z,$ and the set of blossoms $\Omega$ satisfies $(\epsilon_1, \epsilon_2)$-multiplicative approximate complementary slackness if it satisfies Property~\ref{pro:F-approx}(3,4), with Property~\ref{pro:F-approx}(1,2) being replaced with:
\begin{enumerate}[itemsep=0ex]
  \item {\it Approximate Domination.} For each unmatched edge $e \in E \setminus F$, $yz(e) \geq (1 - \epsilon_1) w(e)$.
  \item {\it Approximate Tightness.} For each matched edge $e \in F$, $yz(e) \leq (1 + \epsilon_2)w(e)$.
\end{enumerate}
\end{property}

We can show the following:
\begin{lemma} ~\label{lem:F-mult-approx}
Let $F$ be an $f$-matching in $G$ along with duals $y, z$ and let $F^*$ be the maximum weight $f$-matching. If $F, \Omega, y, z$ satisfy Property~\ref{pro:F-mult-approx} with parameters $\epsilon_1$ and $\epsilon_2$, we have
\[
w(F) \geq (1-\epsilon_1)(1+\epsilon_2)^{-1} w(F^*)
\]
\end{lemma}

We also give the corresponding theorems for $f$-edge covers:

\begin{property}[Approximate Complementary Slackness for $f$-edge cover] ~\label{pro:EC-approx}
Let $\delta_1, \delta_2 \geq 0$ be positive parameters. We say an $f$-edge cover $C$, with duals $y$, $z$ and blossom family $\Omega$ satisfies the $(\delta_1, \delta_2)$-{approximate complementary slackness} if the following requirements holds:
\begin{enumerate}[itemsep=0ex]
  \item {\it Approximate Domination.} For each unmatched edge $e \in E \setminus C$, $yz_C(e) \leq w(e) + \delta_1$.
  \item {\it Approximate Tightness.} For each matched edge $e \in C$, $yz_C(e) \geq w(e) - \delta_2$.
  \item {\it Blossom Maturity.} For each blossom $B \in \Omega$, $|C \cap (\gamma(B) \cup (\delta(B) \setminus I_{C}(B)))| = \left\lceil \frac{f(B) - |I_{C}(B)|}{2} \right\rceil$.
  \item {\it Oversaturated Vertices' Duals.} For each oversaturated vertex $v$, $y(v) = 0$.
\end{enumerate}
\end{property}

\begin{property}[Approximate Complementary Slackness for $f$-edge cover with Multiplicative Error] ~\label{pro:EC-mult-approx}
Let $0 \leq \epsilon_1, \epsilon_2 < 1$ be positive parameters. We say an $f$-edge cover $C$, with duals $y$, $z$ and blossom family $\Omega$ satisfies the $(\epsilon_1, \epsilon_2)$-{approximate complementary slackness} if it satisfies Property~\ref{pro:EC-approx}(3,4), with
Property~\ref{pro:EC-approx}(1,2) being replaced with:
\begin{enumerate}[itemsep=0ex]
  \item {\it Approximate Domination.} For each unmatched edge $e \in E \setminus C$, $yz_C(e) \leq (1 + \epsilon_1)w(e)$.
  \item {\it Approximate Tightness.} For each matched edge $e \in C$, $yz_C(e) \geq (1 - \epsilon_2)w(e)$.
\end{enumerate}
\end{property}

Recall that we are using the aggregated duals $yz_C$ for $f$-edge cover:
\begin{align*}
yz_C(u, v) = y(u) + y(v) + \sum_{B: (u, v) \in \gamma(B) \cup (\delta(B) \setminus I_{C}(B))} z(B)
\end{align*}

\begin{lemma}~\label{lem:ec-cs}
  Let $C$ be an $f$-edge cover with duals $y, z, \Omega$ satisfying Property~\ref{pro:EC-approx} with parameters $\delta_1$ and $\delta_2$, and let $C^*$ be the minimum weight $f$-edge cover. We have $w(C) \leq w(C^*) + \delta_1 |C^*| + \delta_2|C|$.
\end{lemma}

\begin{lemma}~\label{lem:ec-mult-cs}
  Let $C$ be an $f$-edge cover with duals $y, z, \Omega$ satisfying Property~\ref{pro:EC-mult-approx} with parameters $\epsilon_1$ and $\epsilon_2$, and let $C^*$ be the minimum weight $f$-edge cover. We have $w(C) \leq (1 + \epsilon_1)(1 - \epsilon_2)^{-1}w(C^*) $.
\end{lemma}

\section{Connection Between $f$-matchings and $f$-Edge Covers} \label{sec:reduction}

The classical approach to
solve the $f$-edge cover problem is to reduce it to $f$-matching. Specifically, looking for a minimum weight $f_C$-edge cover $C$ for some function $f_C$ can be seen as choosing edges that are \emph{not} in $C$, which is a maximum weight $f_F$-matching where $f_F(u) = \deg(u) - f_C(u)$.

The main drawback of this reduction is that it yields inefficient algorithms.  For example, Gabow's algorithms \cite{Gab18}
for solving maximum weight $f_F$-matching scales linearly with $f_F(V)$, which makes it undesirable when $f_C$ is small.
Even when $f_C(V) = O(n)$, Gabow's algorithm runs in $O(m^2 + mn\log n)$ time.
Moreover, this reduction is not approximation-preserving. In other words, the complement of an \emph{arbitrary}
$(1-\epsilon)$-approximate maximum weight $f_F$-matching is not guaranteed to be a $(1+\epsilon)$-approximate $f_C$-edge cover.

In this section we establish two results. First we prove that a folklore reduction from $1$-edge cover to matching in non-negatively weighted graphs is approximation preserving. 
This allows us to use an efficient approximate matching algorithm for ordinary matching, such as \cite{DP14}, to solve the weighted $1$-edge cover problem. Then we establish the connection between approximate $f_F$-matching and approximate $f_C$-edge cover using approximate complementary slackness from the previous section. This will give a $(1+\epsilon)$-approximate minimum weight $f$-edge cover algorithm from our $(1-\epsilon)$ approximate maximum weight $f$-matching algorithm.

\subsection{Approximate Preserving Reduction from $1$-Edge Cover to $1$-Matchings} \label{sec:reduction-1}

The edge cover problem is a special case of $f$-edge cover where $f$ is $1$ everywhere. The minimum weight edge cover problem is reducible to maximum weight matching, simply by reweighting edges~\cite{Sch02}.
The reduction is as follows:
Let $e(v)$ be any edge with minimum weight incident to $v$ and let $\mu(v) = w(e(v))$. Define a new weight function $w'$ as follows
\[
w'(u, v) = \mu(u) + \mu(v) - w(u, v).
\]
Schrijver \cite[\S 27]{Sch02} showed the following theorem:

\begin{theorem}
Let $M^*$ be a maximum weight matching with respect to a nonnegative weight function $w'$, and $C = M^* \cup \{e(v): v \in V \setminus V(M)\}$. Then $C$ is a minimum weight edge cover with respect to weight function $w$.
\end{theorem}

We show this reduction is also approximation preserving.  Recall that the generally weighted
versions of these problems are reducible to the \emph{non-negatively}
weighted versions in linear time.

\begin{theorem}
Let $M'$ be a $(1-\epsilon)$-maximum weight matching with respect to nonnegative weight function  $w'$, and $C' = M' \cup \{e(v): v \in V \setminus V(M')\}$. Then $C'$ is a $(1+\epsilon)$-minimum weight edge cover with respect to weight function $w$.
\end{theorem}
\begin{proof}
Let $C^*$ and $M^*$ be the optimal edge cover and matching defined previously. By construction, we have
\begin{align*}
  w(C') &= w(M') + \mu(V \setminus V(M')) \\
        &= \mu(V(M')) - w'(M') + \mu(V \setminus V(M')) \\
        &= \mu(V) - w'(M')
\intertext{Similarly, we have $w(C^*) = \mu(V) - w'(M^*)$.
Continuing, we have}
        &\leq \mu(V) - (1-\epsilon) w'(M^*) \\
        &= w(C^*) + \epsilon w'(M^*) \\
        &\leq w(C^*) + \epsilon w(C^*) \\
        &= (1+\epsilon) w(C^*).
\end{align*}
The second to last inequality holds because $M^* \subseteq C^*$
and, by definition,
\[
w'(u, v) = \mu(u) + \mu(v) - w(u, v) \leq 2w(u, v) - w(u, v) = w(u, v).
\]
\end{proof}

The reduction does not naturally extend to $f$-edge cover. In the next section we will show how to obtain a $(1+\epsilon)$-approximate $f$-edge cover algorithm from a $(1-\epsilon)$-approximate $f$-matching within the primal-dual framework.

\subsection{From $f$-edge cover to $f$-matching} \label{sec:reduction-2}

We show that a primal-dual algorithm computing a $(1-\epsilon)$-approximate $f$-matching can be used to compute a $(1+\epsilon)$-approximate $f$-edge cover.  In particular, we show that if we have an $f'$-matching $F$ with blossoms $\Omega$ and duals $y,z$ satisfying Property~\ref{pro:F-approx}, and
an $f$-edge cover $C$ that is $F$'s complement, then the same blossom set $\Omega$ and duals $y$, $z$ can be also used to certify Property~\ref{pro:EC-approx} for $C$ with a same set of parameters. This is formally stated in the following lemma:
\begin{lemma} \label{lem:reduction}
If the duals $y, z, \Omega$ and an $f'$-matching $F$ satisfy Property~\ref{pro:F-approx} with parameters $\delta_1', \delta_2'$,
then the same duals $y, z, \Omega$ and the complementary $f$-edge cover $C = E \setminus F$ satisfy
Property~\ref{pro:EC-approx} with parameters $\delta_1 = \delta_2'$ and $\delta_2 = \delta_1'$.
\end{lemma}

\begin{proof}
It is easy to see a vertex is oversatured in an $f$-edge cover if and only if it is unsaturated in its complementary $f$-matching. Therefore, Property~\ref{pro:EC-approx}(4) (Oversaturated Vertices' Duals) and Property~\ref{pro:F-approx}(4) (Unsaturated Vertices' Duals) are equivalent to each other.

To show Property~\ref{pro:EC-approx}(1) is equivalent to Property~\ref{pro:F-approx}(2), and Property~\ref{pro:EC-approx}(2) is equivalent to Property~\ref{pro:F-approx}(1), it suffices to show that the function $yz_F$ for $f'$-matching $F$  agrees with the function $yz_C$ for its complementary $f$-edge cover $C$. Recall that
\begin{align*}
    &yz_C(u, v) = y(u) + y(v) + \sum_{B: (u, v) \in \gamma(B) \cup (\delta(B) \setminus I_C(B))} z(B), \\
    &yz_F(u, v) = y(u) + y(v) + \sum_{B: (u, v) \in \gamma(B) \cup I_F(B)} z(B).
\end{align*}

\begin{figure}
  \centering
  \includegraphics[scale=1.2]{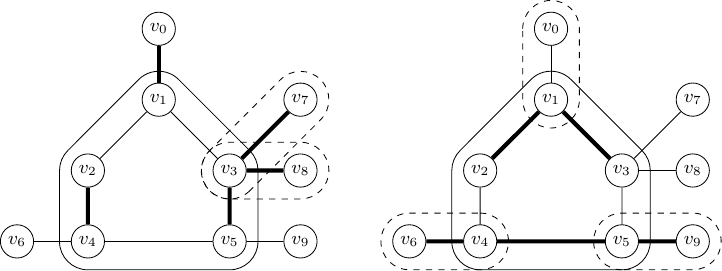}
  \caption{Illustration on relation between $I$-set of an $f$-matching and the $I$-set of its complementary $f'$-edge cover. Left: an $f$-matching and its blossom set. Right: Its complementary $f$-edge cover. Their $I$-sets are circled (dashed).}\label{fig:blossom-3}
\end{figure}

Here $I_C$ and $I_F$ refer to the $I$-sets of a blossom with respect to the
$f$-edge cover $C$ and the $f'$-matching $F$. This reduces to showing that $I_F(B) = \delta(B) \setminus I_C(B)$:
\begin{align*}
  I_F(B) &= \eta(B) \oplus \delta_F(B) = \eta(B) \oplus (\delta(B) \oplus \delta_C(B)) \\
       &= \delta(B) \oplus (\eta(B) \oplus \delta_C(B)) = \delta(B) \setminus I_C(B).
\end{align*}
Therefore, in $yz_F(e)$ and $yz_C(e)$, $z$-values are summed up over the same set of blossoms in $\Omega$. In other words, $yz_F(e) = yz_C(e)$ for each $e \in E$ and the claim follows

To prove that Property~\ref{pro:F-approx}(3) implies Property~\ref{pro:EC-approx}(3), we argue by definition that the maturity of an $f'$-matching blossom implies the maturity of the corresponding $f$-edge-cover blossom. Equality is then implied by Lemma~\ref{lem:mature-factor} and Lemma~\ref{lem:mature-cover}. Indeed, by how we define our $f'$-matching $F$ and $f$-edge cover $C$, a vertex's surplus with respect to $C$ and $f$ is equal to a vertex's deficiency with respect to $F$ and $f'$. Moreover, the blossom is heavy/light for $f'$-matching iff it is light/heavy for the corresponding $f$-edge cover. Since the base edge is defined to be the same, maturity of one blossom implies the other. This completes the proof.
\end{proof}
\section{Approximation Algorithms for $f$-Matching and $f$-Edge Cover} \label{sec:approx}

In this section, we prove the main result by giving an approximation algorithm for computing $(1-\epsilon)$-approximate maximum weight $f$-matching. The crux of the result is an implementation of Edmonds' search with relaxed complementary slackness as the eligibility criterion. The notion of approximate complementary slackness was introduced by Gabow and Tarjan for both bipartite matching \cite{GT89} and general matching \cite{GT91}. Gabow gave an implementation of Edmonds' search with \emph{exact} complementary slackness for the $f$-matching problem  \cite{Gab18}, which finds augmenting walks one at a time. The main contribution of this section is to adapt  \cite{Gab18} to approximate complementary slackness to facilitate finding augmenting walks in batches.

To illustrate how this works, we will first give an approximation algorithm for $f$-matching in graphs with small edge weights. Let $w(\cdot)$ be a positive weight function $w: E \rightarrow \{0, \ldots, W\}$.  The algorithm computes a $(1-\epsilon)$-approximate maximum weight $f$-matching in $O(m\alpha(m, n)W\epsilon^{-1})$ time, independent of $f$. We also show how to use scaling techniques to transform this algorithm to run in $O(m\alpha(m, n)\epsilon^{-1} \log \epsilon^{-1})$ time, independent of $W$.

\subsection{Approximation for Small Weights} \label{sec:small-weight}

The main procedure in our $O(m\alpha(m, n)W\epsilon^{-1})$ time algorithm is a variation on Edmonds' search. In one iteration, Edmonds' search finds a set of augmenting walks using \emph{eligible edges}, creates and dissolves blossoms, and performs \emph{dual adjustments} on $y$ and $z$ while maintaining the following Invariant:
\begin{invariant}[Approximate Complementary Slackness] \label{inv:F-approx}
Let $\delta > 0$ be some parameter such that $w(e)$ is a multiple of $\delta$, for all $e \in E$:
\begin{enumerate}[itemsep=0ex]
  \item {\it Granularity.} $y$-values are multiples of $\delta/2$ and $z$-values are multiples of $\delta$.
  \item {\it Approximate Domination.} For each unmatched edge and each blossom edge $e \in (E \setminus F) \cup (\bigcup_{B \in \Omega} E_B)$, $yz(e) \geq w(e) - \delta$.
  \item {\it Approximate Tightness.} For each matched and each blossom edge $e \in F \cup (\bigcup_{B \in \Omega} E_B) $, $yz(e) \leq w(e)$.

  \item {\it Blossom Maturity.} For each blossom $B \in \Omega$, $|F \cap (\gamma(B) \cup I(B))| = \left\lfloor \frac{f(B) + |I(B)|}{2}\right\rfloor$. Root blossoms in $\Omega$ have positive $z$-values.
  \item {\it Unsaturated Vertices.} All unsaturated vertices have the same $y$-value; their $y$-values are strictly less than the $y$-values of other vertices.
\end{enumerate}
\end{invariant}

Notice that here we relax Property~\ref{pro:F-approx}(4) to allow unsaturated vertices to have positive $y$-values. The purpose of Edmonds' search is to decrease the $y$-values for all unsaturated vertices while maintaining Invariant~\ref{inv:F-approx}. Following \cite{GT91,DP14,DPS18}, we define the following eligibility criterion:

\begin{criterion} \label{cri:approx}
An edge $(u, v)$ is \emph{eligible} if it satisfies one of the following:
\begin{enumerate}[itemsep=0ex]
  \item $e \in E_B$ for some $B \in \Omega$.
  \item $e \not\in F$ and $yz(e) = w(e) - \delta$.
  \item $e \in F$ and $yz(e) = w(e)$.
\end{enumerate}
\end{criterion}

A key property of this definition is that it is asymmetric for matched and unmatched edges that are not in any blossom. As a result, if we augment along an eligible augmenting walk $P$, all edges in $P$, except for those in contracted blossoms, will become ineligible, 
and $P$'s image in the contracted graph will become entirely ineligible.

We define $\Gelig$ to be the graph obtained from $G$ by discarding all ineligible edges, and let $\widehat{\Gelig} = \Gelig / \Omega$ be obtained from $\Gelig$ by contracting all blossoms in $\Omega$. For initialization, we set $F = \emptyset$, $y = W/2$, $z = 0$, $\Omega = \emptyset$. Edmonds' search repeatedly executes the following steps: \emph{Augmentation}, \emph{Blossom Formation}, \emph{Dual Adjustment}, and \emph{Blossom Dissolution} until all unsaturated vertices 
have $y$-values equal to zero.  See Figure~\ref{fig:approx-f-factor-noscaling}.

\begin{figure}
\begin{framed}
\begin{enumerate}[itemsep=0pt]
  \item {\it Augmentation.} Find a set
  of edge-disjoint augmenting walks $\widehat{\Psi}$ and a set of alternating cycles $\widehat{\mathcal{C}}$ in $\widehat{\Gelig}$, such that after removing their edges from $\widehat{\Gelig}$, $\widehat{\Gelig}$ does not contain any augmenting walk.
    Let $\Psi$ and $\mathcal{C}$ be the preimages of $\widehat{\Psi}$ and $\widehat{\mathcal{C}}$ in $\Gelig$.
    Update $\displaystyle F \leftarrow F \oplus \left(\left(\bigcup_{P \in \Psi}P\right) \cup \left(\bigcup_{C \in \mathcal{C}}C\right)\right)$.
    After this step, the new $\widehat{\Gelig}$ contains no
    augmenting walk.
  \item {\it Blossom Formation.} Find a maximal set $\Omega'$ of nested blossoms reachable from an unsaturated vertex/blossom in $\widehat{\Gelig}$.
  Update $\Omega \leftarrow \Omega \cup \Omega'$ and then update
  $\widehat{\Gelig}$ to be $G/\Omega$.
  After this step, $\widehat{\Gelig}$ contains no blossom
  reachable from an unsaturated vertex/blossom.
  \item {\it Dual Adjustment.} Let $\widehat{S}$ be the set of vertices from $\widehat{\Gelig}$ reachable from an unsaturated vertex via an eligible alternating walk. We classify vertices in $\widehat{S}$ into $\widehat{V_{in}}$, the set of inner vertices and $\widehat{V_{out}}$, the set of outer vertices.\footnotemark {} Let $V_{in}$ and $V_{out}$ be the set of original vertices in $V$ represented by $\widehat{V_{in}}$ and $\widehat{V_{out}}$. Adjust the $y$ and $z$ values as follows:
      \begin{align*}
        y(v) &\leftarrow y(v) - \delta/2, \mbox{ if $v \in V_{out}$} \\
        y(v) &\leftarrow y(v) + \delta/2, \mbox{ if $v \in V_{in}$} \\
        z(B) &\leftarrow z(B) + \delta,   \mbox{ if $B$ is a root blossom in $\widehat{V}_{out}$} \\
        z(B) &\leftarrow z(B) - \delta,   \mbox{ if $B$ is a root blossom in $\widehat{V}_{in}$}
    \end{align*}
  Here a root blossom is an inclusionwise maximal blossom in $\Omega$.
  \item {\it Blossom Dissolution.} After Dual Adjustment some root blossoms in $\Omega$ might have $0$ $z$-values. Remove them from $\Omega$ until none exists. Update $\Omega$ and $\widehat{\Gelig}$.
\end{enumerate}
\end{framed}
\caption{\label{fig:approx-f-factor-noscaling} A $(1-\epsilon)$-approximate $f$-matching algorithm for small integer weights.}
\end{figure}

\footnotetext{In an actual implementation, the inner/outer labelling can be computed in the search in Blossom Formation step. The labelling continues to be valid after contracting a maximal set of blossoms.}

Now we define what we mean by \emph{reachable} vertices in Steps 1--3 of the algorithm, as well as the inner/outer labelling of nontrivial blossoms and singletons. This is analogous to the reachable/inner/outer vertices in Edmonds' Search for ordinary matching~\cite{DP14,DPS18}, except that we cannot simply treat a contracted blossom like a single vertex. The corresponding definition for $f$-matching is given in Gabow~\cite[p. 46]{Gab18}. For completeness, we restate these definitions and further supplement them with the notion of \emph{alternation}, which provides further insights for reachability.

We start by defining \emph{alternation} which follows from Definition~\ref{def:AP-f-factor} of an augmenting walk.
We say two distinct edges $e,e'$ incident to a blossom/singleton $B$ \emph{alternate} if either $B$ is a singleton and $e$ and $e'$ are of different types,
or $B$ is a nontrivial blossom and $|\eta(B) \cap \{e, e'\}| = 1$.
An \emph{alternating walk} in the contracted graph is a walk where every two consecutive edges alternates.
An augmenting walk is an alternating walk with its terminal edges and terminal vertices satisfying the requirement specified in Definition~\ref{def:AP-f-factor}.

$\widehat{S}$ is the set of blossoms and vertices in $\widehat{\Gelig}$ that are reachable from an unsaturated singleton or an unsaturated light blossom  via an eligible alternating walk. It can be obtained by inductively constructing an alternating search tree rooted at an unsaturated singleton or an unsaturated light blossom.
We label the root nodes \emph{outer}. For a nonroot vertex $v$ in $\widehat{S}$ , let $\tau(v)$ be the edge in $\widehat{S}$ pointing to the parent of $v$. The inner/outer status of $v$ is defined as follows:

\begin{definition}\cite[p. 46]{Gab18} \label{def:in-out}
A vertex $v$ is \emph{outer} if one of the following is satisfied:
\begin{enumerate}[itemsep=0ex]
  \item $v$ is the root of a search tree.
  \item $v$ is a singleton and $\tau(v) \in F$.
  \item $v$ is a nontrivial blossom and $\{\tau(v)\} = \eta(v)$.
\end{enumerate}

Otherwise, one the the following holds and $v$ is classified as \emph{inner}:
\begin{enumerate}[itemsep=0ex]
  \item $v$ is a singleton and $\tau(v) \in E\backslash F$.
  \item $v$ is a nontrivial blossom and $\{\tau(v)\} \neq \eta(v)$.
\end{enumerate}
\end{definition}

An individual search tree in $\widehat{S}$, call it $\widehat{T}$, can be grown by repeatedly attaching a child $v$ to its parent $u$ using an edge $(u, v)$ that is \emph{eligible for} $u$ in $\widehat{S}$; 
see Gabow \cite[p. 46]{Gab18}. Let $B_u$ denote the root blossom in $\Omega$ containing $u$. We say an edge $(u, v) \in E$ is \emph{eligible for} $u$ if it is eligible according to Criterion~\ref{cri:approx} and one of the following is satisfied:
\begin{enumerate}[itemsep=0ex]
  \item $u$ is an outer singleton and $e \not\in F$.
  \item $B_u$ is an outer blossom and $\{e\} \neq \eta(B_u)$.
  \item $u$ is an inner singleton and $e \in F$.
  \item $B_u$ is an inner blossom and $\{e\} = \eta(B_u)$.
\end{enumerate}

Hence, $\widehat{S}$ consists of singletons and blossoms that are reachable from an unsaturated singleton or light blossom, via an eligible alternating path. We call such blossoms and singletons \emph{reachable}, and all other singletons and blossoms \emph{unreachable}. A vertex $v$ from the original graph $\Gelig$ is reachable (unreachable) if $B_v$ is reachable (unreachable) in $\widehat{\Gelig}$.\footnote{Of course, if $B_v$ is inner and reachable in
$\widehat{G}$, this only implies that $\beta(B_v)$
is reachable from an unsaturated vertex in $G$;
other vertices in $B_v$ may not be reachable
in $G$.}

\begin{figure}
  \centering
  \includegraphics[scale=1.5]{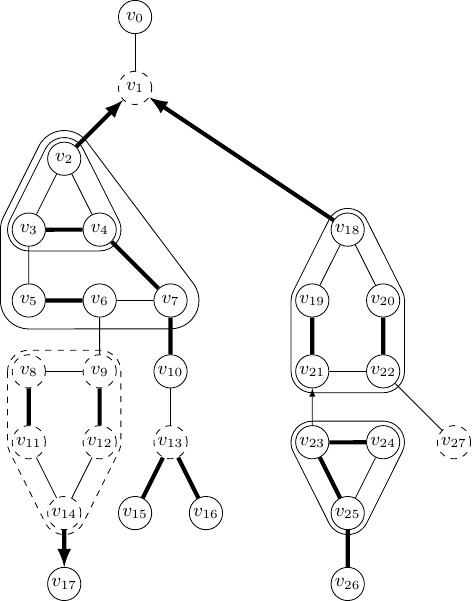}
  \caption{An example of an eligible alternating search tree. Outer blossoms and singletons are labeled using solid boundaries while inner blossoms and singletons have dashed boundaries. }\label{fig:edmonds-2}
\end{figure}

In Edmonds' Search, primal and dual variables are initialized in a way that Property~\ref{pro:F-approx}(1) (Approximate Domination) is always satisfied, and Property~\ref{pro:F-approx} (Approximate Tightness) is vacuous (as the $f$-matching is initially empty) but Property~\ref{pro:F-approx}(4) (Unsaturated Vertices) is not. For this reason, there is a large gap between primal and dual objective, besides the error introduced by \emph{approximate} tightness and domination, at the beginning of the algorithm. This gap is given by the following, assuming exact tightness and domination is satisfied (i.e. $\delta_1 = \delta_2 = 0$ in Property~\ref{pro:F-approx}):

\begin{align*}
  yz(V) - w(F) &= \sum_{v \in V} f(v) y(v) + \sum_{B \in \Omega} \left\lfloor \frac{f(B) + |I(B)|}{2}\right\rfloor z(B) + \sum_{e \in E} u(e) - \sum_{e \in F} w(e) \\
               &= \sum_{v \in V} \defi(v) y(v).
\end{align*}

The goal of the algorithm can be seen as bridging the gap between the primal objective and dual objective while preserving all other complementary slackness properties. It can be achieved in two ways. Augmentations enlarge the $f$-matching by augmenting $F$ along some augmenting walk $P$. This will reduce the total deficiency on the vertex set $V$. Dual Adjustments change the dual variables in a way that decreases the $y$-values on unsaturated vertices while maintaining other approximate complementary slackness conditions. In this algorithm, the progress of Edmonds' Search is measured by the latter, i.e., the overall reduction in $y$-values of unsaturated vertices.

The correctness of our algorithm reduces to showing that \emph{Augmentation}, \emph{Blossom Formation}, \emph{Blossom Dissolution},
and \emph{Dual Adjustment} all preserve Invariant~\ref{inv:F-approx}.

\begin{lemma} \label{lem:ag-bf-bd}
The Augmentation step, Blossom Formation step, and Blossom Dissolution step 
preserve Invariant~\ref{inv:F-approx}.
\end{lemma}

\begin{proof}
We first show that the identity of $I(B)$ is invariant under an augmentation; in particular, augmenting along an augmenting walk that intersects $B$ does not change $I(B)$.
As a result, the function $yz(\cdot)$ is invariant under augmentation. This is a restatement of Lemma 5.3 in \cite{Gab18}, for completeness, we restate the proof.

We use $I(B), \eta(B)$ and $I'(B), \eta'(B)$ to denote the $I$-set and base edge of $B$ before and after the augmentation. By Definition~\ref{def:AP-f-factor} (augmenting walks), if $P$ intersects $B$, then
\begin{equation*}
   \delta_P(B) = \eta(B) \cup \eta'(B) = \eta(B) \oplus \eta'(B).
\end{equation*}
Let $F$ and $F'$ be the $f$-matching before and after augmentation.  We have
\begin{equation*}
  \delta_{F'}(B) = \delta_F(B) \oplus \delta_P(B)
\end{equation*}
Combining both equations, we have
\begin{equation*}
  \delta_{F'}(B) = \delta_F(B) \oplus (\eta(B) \oplus \eta'(B))
\end{equation*}
Hence
\begin{equation*}
  I'(B) = \delta_{F'}(B) \oplus \eta'(B) = \delta_F(B) \oplus \eta(B) = I(B).
\end{equation*}

By Invariant~\ref{inv:F-approx}, any blossom edge $e \in \bigcup_{B \in \Omega} E_B$ satisfies both approximate domination as well as approximate tightness, so it continues to satisfy these Invariants after augmentation. For any eligible edge not in $E_B$ for any $B \in \Omega$, by Criterion~\ref{cri:approx}, if $e$ is matched,  $yz(e) = w(e) - \delta$, thus after the Augmentation step its duals satisfy approximate domination. If $e$ is unmatched, $yz(e) = w(e)$, so its duals satisfy approximate tightness after the Augmentation step.

Augmentation also preserves the maturity of blossoms. For any vertex $v$ in a nonterminal blossom $B$, $\deg_F(v) = \deg_{F'}(v) = f(v)$, so maturity is naturally preserved. If $B$ is a terminal blossom, we have $\deg_F(v) = f(v) - 1$ for $v = \beta(B)$ and $\deg_F(v) = f(v)$ for all $v \neq \beta(B)$. Moreover, after Augmentation $B$ always has a base edge $\eta(B) = \delta_P(B)$. 
Therefore, $B$ is also mature after Augmentation.

All the blossoms found in Blossom Formation must be mature and have $z$-values equal to zero, 
so the value of the $yz$ function is unchanged and all the invariants are preserved.

For Blossom Dissolution step, discarding blossoms with zero $z$-values preserves the value of the $yz$ function and hence preserves the invariants.
\end{proof}

The crux of the proof is to show that Dual Adjustment also preserves Invariant~\ref{inv:F-approx}, in particular approximate domination and approximate tightness. Before proving the correctness of Dual Adjustment, we first prove the following parity lemma, which was first used in \cite{GT91}; we generalize it to $f$-matching:

\begin{lemma}[Parity] \label{lem:parity}
Let $\widehat{S}$ be the search forest defined as above. Let $S$ be the preimage of $\widehat{S}$ in $G$. The $y$-value of every vertex in $S$ has the same parity, as a multiple of $\delta/2$.
\end{lemma}

\begin{proof}
The claim clearly holds after initialization as all vertices have the same $y$-values.
Now notice that every \emph{eligible edges} $e = (u, v)$ that straddles two distinct singletons or nontrivial blossom must have its $yz$-value being $w(e)$ or $w(e) - \delta$. Since $w(e)$ is by assumption an integral multiple of $\delta$, $yz(e)$ is also a multiple of $\delta$. Because $z$-values are always multiples of $\delta$, $y(u)$ and $y(v)$ must both be odd or even as a multiple of $\delta/2$.

Therefore it suffices to show that every vertex in a blossom $B \in \Omega$ has the same parity.

To prove this, we only need to show that the Blossom Formation step only groups vertices with the same parity together. This is because new blossoms $B$ are formed when we encounter a cycle of nontrivial blossoms and singletons $C_B = \left<B_0, e_0, B_1, e_1, \ldots, B_{l-1}, e_{l-1}\right>$ whose edges are eligible. Therefore the endpoints of those edges share the same parity. Hence by induction, all vertices in $B$ also share the same parity. The Dual Adjustment step also preserves this property as vertices in a blossom will have the same inner/outer classification and thus have their $y$-values all incremented or decremented by $\delta/2$.
\end{proof}

The following theorem is a generalization of Lemma~5.8 in \cite{Gab18} to approximate complementary slackness. The proof follows from the same framework but has a slightly more complicated case analysis.

\begin{lemma}\label{lem:dual-adjustment}
Dual Adjustment and Blossom Dissolution preserve Invariant~\ref{inv:F-approx}.
\end{lemma}

\begin{proof}
We focus on part 2 (Approximate Domination) and part 3 (Tightness) of Property~\ref{pro:F-approx}. Part 1 (Granularity) is naturally preserved since we are adjusting $y$-values by $\delta/2$ and $z$-values by $\delta$. Part 5 (Unsaturated vertices duals) is also preserved because unsaturated vertices are labelled as outer and their dual is adjusted by the same amount. Maturity of blossoms is not affected by Dual Adjustment. Although after dual adjustment, some (inner) root blossoms might have $0$ $z$-values, such blossoms are removed in Blossom Dissolution step so part 4 for Invariant~\ref{inv:F-approx} is restored at the end of the iteration.

Similar to ordinary matching, preservation of approximate domination and tightness can be argued using a case analysis on vertices' and blossoms' duals. 
Notice that there are more cases to consider in $f$-matching compared to ordinary matching. Different cases can be generated for an edge $(u,v)$ by considering the inner/outer classification of both endpoints, whether $(u, v)$ is matched, whether $(u, v)$ is the base edge for its respective endpoints, if they are in blossoms, and  whether $(u, v)$ is eligible.
In the following analysis, we follow the framework in Lemma 5.8 \cite{Gab18} to narrow down the number of meaningful cases to just $8$. Notice that Lemma 5.8 \cite{Gab18} can be seen as a version of this lemma for exact complementary slackness. Although one can expect the same conclusion to hold, the proof still differs in the details.

We consider an edge $e = (u, v)$. If $u$ and $v$ are both unreachable, or both in the same root blossom, $yz(u, v)$ clearly remains unchanged after Dual Adjustment.

Therefore we can assume $B_u \neq B_v$ and at least one of them, say $B_u$, is reachable. Every reachable endpoint will contribute a change of $\pm \delta/2$ to $yz(u, v)$. This is the adjustment of $y(u)$, plus the adjustment of $z(B_u)$ if $e \in I(B_u)$. Define $\Delta(u)$ to be the net change of the quantity $y(u) + \sum_{e \in I(B_u)} z(B_u)$.
By definition of Dual Adjustment, we have the following scenarios:
\begin{itemize}[itemsep=0pt]
  \item $\Delta(u) = +\delta/2$: This occurs if $u$ is an inner singleton, or $B_u$ is an outer blossom with $e \in I(B_u)$, or an inner blossom with $e \not\in I(B_u)$.
  \item $\Delta(u) = -\delta/2$: This occurs if $u$ is an outer singleton, or $B_u$ is an inner blossom with $e \in I(B_u)$, or an outer blossom with $e \not\in I(B_u)$.
\end{itemize}

Then we consider the effect of a Dual Adjustment on the edge $e = (u, v)$. First we consider the case when exactly one of $B_u$ and $B_v$, say $B_u$, is in $\widehat{S}$. In this case only $u$ will introduce a change on $yz(u, v)$:

\case{1} $u$ is an inner singleton: Here $\Delta(u) = +\delta/2$. In this case approximate domination is preserved, so we only need to worry about approximate tightness and hence assume $e \in F$. Since $B_v$ is not in $\widehat{S}$, $e$ cannot be eligible for $B_u$ or $B_v$ would have been included in $\widehat{S}$ as a child of $B_u$. Because $e \in F$, $e$ cannot be eligible. Hence $yz(e)<w(e)$. By Granularity, $yz(e) \leq w(e) - \delta/2$. Therefore we have $yz(e) \leq w(e)$ after the Dual Adjustment.

\case{2} $u$ is an outer singleton: Here $\Delta(u) = -\delta/2$. In this case tightness is preserved and we only need to worry about approximate dominination when $e \not\in F$. Similar to Case 1, $e$ must be ineligible and $yz(e) \geq w(e) - \delta/2$. After Dual Adjustment we have $yz(e) \geq w(e) - \delta$.

\case{3} $B_u$ is an inner blossom: We divide the cases according to whether $e$ is matched or not.

\subcase{3.1} $e \in F$. If $e \not\in  \eta(B_u)$, then $e \in I(B_u)$ and $\Delta(u) = -\delta/2$. In this case tightness is preserved. If $e \in \eta(B_u)$, then $e \not\in I(B_u)$ and $\Delta(u) = +\delta/2$. But $e$ cannot be eligible since otherwise $B_v$ would be in the search tree, so we have $yz(e) \leq w(e) - \delta/2$ and $yz(e) \leq w(e)$ after Dual Adjustment.

\subcase{3.2} $e \not\in F$. This is basically symmetric to Subcase 3.1. If $e \in \eta(B_u)$, then $e \in I(B_u)$ and $\Delta(u) = -\delta/2$. But $e$ cannot be eligible therefore $yz(e) \geq w(e) - \delta/2$, and $yz(e) \geq w(e) - \delta$ after Dual Adjustment. If $e \not\in \eta(B_u)$, then $e \not \in I(B_u)$ and $\Delta(u) = +\delta/2$, so approximate Domination is preserved.

\case{4} $B_u$ is an outer blossom:

\subcase{4.1} $e \in F$. If $e \in \eta(B_u)$, then $B_v$ must be  the parent of $B_u$ in the search tree, contradicting the fact that $B_v \not\in \widehat{S}$. Thus $e \not\in \eta(B_u)$, so $e \in I(B_u)$ and $\Delta(u) = +\delta/2$. Since $B_v$ is not reachable, $e$ cannot be eligible, so $yz(u, v) \leq w(e) - \delta/2$ before Dual Adjustment and $yz(u, v) \leq w(e)$ afterward.

\subcase{4.2} $e \not\in F$. Similarly, $e \not\in \eta(B_u)$, so $e \not\in I(B_u)$ and $\Delta(u) = -\delta/2$. Similarly $B_v$ is not reachable so $e$ cannot be eligible. Therefore we have $yz(u, v) \geq w(e) - \delta/2$ and $yz(u, v) \geq w(e) - \delta$ after Dual Adjustment.

\medskip
This completes the case when exactly one of $e$'s endpoints is reachable. The following part will complete the argument for when both endpoints are reachable. We argue that three scenarios can happen: either $\Delta(u)$ and $\Delta(v)$ are of opposite signs and cancel each other out, or $\Delta(u)$ and $\Delta(v)$ are of the same sign and the sign aligns with the property we wish to keep, or if neither case holds, we use Lemma~\ref{lem:parity}~(Parity) to argue that there is enough room for dual adjustment not to violate approximate domination or tightness.

We first examine tree edges in $\widehat{S}$. In this case we assume $B_u$ is the parent of $B_v$ and $e$ is the parent edge of $B_v$. Hence $e$ must be eligible for $B_u$. We argue by the sign of $\Delta(u)$.

\case{5} If $e$ is a tree edge and $\Delta(u) = +\delta/2$:

There are three cases here: $u$ is an inner singleton, $B_u$ is an outer blossom with $e \in I(B_u)$, or $B_u$ is an inner blossom with $e \not\in I(B_u)$. We first observe that in all three cases, $e \in F$. This is straightforward when $u$ is an inner singleton. If $B_u$ is an outer blossom with $e \in I(B_u)$, we know that since $B_u$ is outer, $e \not\in \eta(B_u)$, so therefore $e \in F$. If $B_u$ is an inner singleton with $e \not\in I(B_u)$, since $B_u$ is inner, $e \in \eta(B_u)$, so combined with the fact that $e \not\in I(B_u)$ we have $e \in F$.

Notice that since $B_u$ is the parent of $B_v$, and $e \in F$, $v$ can be an outer singleton, or $B_v$ is an outer blossom with $e \in \eta(B_v)$, or $B_v$ is an inner blossom with $e \not\in \eta(B_v)$. In the second case $e \not\in I(B_v)$ and in the third case $e \in I(B_v)$. In all three cases we have $\Delta(v) = -\delta/2$, and $yz(e)$ remains unchanged.

\case{6} If $e$ is a tree edge and $\Delta(u) = -\delta/2$: Case 6 is symmetric to Case 5. $B_u$ can either be an outer singleton, an inner blossom with $e \in I(B_u)$, or an outer blossom with $e \not\in I(B_u)$. In all cases, the fact that $e$ must be eligible for $B_u$ implies $e \not\in F$, and $B_v$ can only be an inner singleton, an outer blossom with $e \in I(B_u)$, or an inner blossom with $e \not\in I(B_v)$. Hence we have $\Delta(v) = +\delta/2$ so $yz(e)$ remains unchanged.

\medskip
Now suppose $B_u$ and $B_v$ are both in $\widehat{S}$ but $(u, v)$ is not a tree edge. We still break the cases according to the sign of $\Delta(u)$ and $\Delta(v)$. Here we only need to consider when $\Delta(u) = \Delta(v)$, since otherwise they cancel each other and $yz(e)$ remains constant.

\case{7} If $e$ is a tree edge and $\Delta(u) = \Delta(v) = \delta/2$. In this case $yz(e)$ is incremented by $\delta$. Therefore we only need to worry about tightness when $e \in F$. Notice that $B_u$ can only be an inner singleton, an outer blossom with $e \in I(B_u)$ or an inner blossom with $e \not\in I(B_u)$. When $B_u$ is an outer blossom, $e \not\in \eta(B_u)$. When $B_u$ is an inner blossom, since $e \in F$ and $e \not\in I(B_u)$, $e \in \eta(B_u)$. The same holds for the other endpoint $B_v$.

It is easy to verify that in all cases, $e$ is eligible for $B_u$ (or $B_v$) if and only if $e$ is eligible. But notice that after Augmentation and Blossom Formation steps, there is no augmenting walk or reachable blossom in $\widehat{\Gelig}$, i.e., there cannot be an edge $(u, v)$ that is eligible for both endpoints $B_u$ and $B_v$ since otherwise one can find an augmenting walk or a new reachable blossom. Thus $e$ is ineligible and $yz(e) < w(e)$. But by Invariant~\ref{inv:F-approx}(1) (Granularity) and Lemma~\ref{lem:parity} (Parity), both $w(e)$ and $yz(e)$ must be multiples of $\delta$. Therefore we have $yz(e) \leq w(e) - \delta$. This implies $yz(e) \leq w(e)$ after Dual Adjustment.

\case{8} If $e$ is a tree edge and $\Delta(u) = \Delta(v) = -\delta/2$. Here $yz(e)$ is decremented by $\delta$. Similar to the case above, we can assume $e \not\in F$ and only focus on approximate domination. $B_u$ can be an outer singleton, inner blossom with $e \in I(B_u)$, or outer blossom with $e \not\in I(B_u)$. Since $e \not\in F$, $e \in I(B_u)$ if and only if $e \in \eta(B_u)$. Therefore, if $e$ is eligible, $e$ must be eligible for both $B_u$ and $B_v$. But similar to Case 7, $e$ being eligible for both endpoints will lead to the discovery of an additional blossom or augmenting walk in $\Gelig$, which is impossible after Augmentation and Blossom Formation. Therefore we conclude in this case $e$ is ineligible and $yz(e) > w(e) - \delta$. By Lemma~\ref{lem:parity} (Parity), we have $yz(e) \geq w(e)$ before Dual Adjustment, so approximate domination still holds after Dual Adjustment.
\end{proof}

\begin{theorem} ~\label{thm:small-weight}
A $(1-\epsilon)$-approximate $f$-matching can be computed in $O(Wm\alpha(m, n)\epsilon^{-1})$ time.
\end{theorem}
\begin{proof}
We initialize the $f$-matching to be $\emptyset$ and $y(v) = W/2$ for all $v$. Set $\delta = 1/\left\lceil \epsilon^{-1} \right\rceil \leq \epsilon$. Since each iteration decreases $y$-values by $\delta/2$, $y$-values of unsaturated vertices takes $(W/2)/(\delta/2) = O(W\epsilon^{-1})$ iterations to reach $0$, thereby satisfying Property~\ref{pro:F-approx} with $\delta_1 = \delta, \delta_2 = 0$. By invoking Lemma~\ref{lem:F-approx}, with $F^*$ being the optimum $f$-matching, we have
\begin{equation*}
  w(F) \geq w(F^*) - |F^*| \delta \geq w(F^*) - w(F^*) \delta \geq (1 - \epsilon) w(F^*).
\end{equation*}

For the running time, each iteration of Augmentation, Blossom Formation, Dual Adjustment, and Blossom Dissolution can be implemented in $O(m \alpha(m, n))$ time. 
We defer the detailed implementation to Section~\ref{sec:aug-walk}. There are a total of $W/\delta = O(W\epsilon^{-1})$ iterations, so the running time is $O(Wm\alpha(m, n)\epsilon^{-1})$.
\end{proof}

As a result of Lemma~\ref{lem:reduction} and Lemma~\ref{lem:ec-cs}, we also obtain the following result:

\begin{corollary}~\label{cor:small-weight}
A $(1+\epsilon)$-approximate $f$-edge cover can be computed in $O(Wm\alpha(m, n)\epsilon^{-1})$ time.
\end{corollary}
\begin{proof}
Given a weighted graph $G$ and degree constraint function $f$, let $f' = \deg - f$ be the complement of $f$. With some paramter $\delta$ we run the algorithm from Theorem~\ref{thm:small-weight} to find an $f'$-matching $F'$ that satisfies Property~\ref{pro:F-approx} with parameters $(\delta, 0)$. By Lemma~\ref{lem:reduction}, its complement $F = E \setminus F'$ satisfies Property~\ref{pro:EC-approx} with parameters $(0, \delta)$. By Lemma~\ref{lem:ec-cs}, we have
\begin{align*}
w(F^*) \geq w(F) - \delta|F| \geq (1-\delta)w(F)
\end{align*}
Then we can choose $\delta = \Theta(\epsilon)$ to guarantee that $F$ is 
an $(1+\epsilon)$-approximate minimum weight $f$-edge cover.
\end{proof}

Also notice that when $W = O(1)$ is constant, Theorem~\ref{thm:small-weight} and Corollary~\ref{cor:small-weight} are the fastest known approximation algorithms for these problems.

\subsection{A Scaling Algorithm for General Weights}

In this section, we can modify the $O(W m\alpha(m, n)\epsilon^{-1})$ weighted $f$-matching algorithm to work on graphs with general real weights. The modification is based on the scaling framework in \cite{DP14}.
If the weights are arbitrary reals, we can round
them to integers in $[W]$, $W=\mbox{poly}(n)$, with negligible loss in accuracy.  Thus we can assume without loss of generality that all weights are $O(\log n)$-bit integers.
The idea is to divide the algorithm into into $L = \log W + 1$ scales that execute Edmonds' search with exponentially diminishing $\delta$. The goal of each scale is to use $O(\epsilon^{-1})$ Edmonds' searches to halve the $y$-values of all unsaturated vertices while maintaining a more relaxed version of approximate complementary slackness. By manipulating the weight function, approximate domination, which is weak at the beginning, is strengthened over scales, while approximate tightness is weakened in exchange.  Assume without loss of generality that $W > 1$ and $\epsilon < 1$ are powers of two.  We define $\delta_i, 0 \leq i \leq L$ to be the error parameter in scale $i$, 
where $\delta_0 = \epsilon W$ and $\delta_i = \delta_{i - 1} / 2$ for $0 < i \leq L$.
Each scale works with a new weight function $w_i$ which is the old weight function rounded down to the nearest multiple of $\delta_i$, i.e, $w_i(e) = \delta_i \left\lfloor w(e)/\delta_i \right\rfloor$.
In the last scale $W_L = w$. We maintain a scaled version of Invariant~\ref{inv:F-approx}.
Note differences in the Approximate Domination and Approximate Tightness criteria.

\begin{invariant}[Scaled approximate complementary slackness with positive unsaturated vertices] \label{inv:F-scale}
At scale $i = 0, 1, \ldots, L = \log W$, we maintain the $f$-matching $F$, blossoms $\Omega$, and duals $y$, $z$ to satisfy the following invariant:
\begin{enumerate}[itemsep=0pt]
  \item {\it Granularity.} All $y$-values are multiples of $\delta_i/2$, and $z$-values are multiples of $\delta_i$.
  \item {\it Approximate Domination.} For each $e \not\in F$ or $e \in E_B$ for some $B \in \Omega$, $yz(e) \geq w_i(e) - \delta_i$.
  \item {\it Approximate Tightness.} For each
  $e \in F\cup (\bigcup_{B\in\Omega} E_B)$,
  let $j_e \leq i$ be the index of last scale that $e$
  joined the set $F \cup \bigcup_{B\in\Omega} E_B$.
  We have
  $yz(e) \leq w_i(e) + 2 \delta_{j_e} - 2 \delta_i$.
  \item {\it Mature Blossoms.} For each blossom $B \in \Omega$, $|F \cap (\gamma(B) \cup I(B))| = \left\lfloor \frac{f(B)+ |I(B)|}{2} \right\rfloor$.
  \item {\it Unsaturated Vertices' Duals.} The $y$-values of all unsaturated vertices are the same and less than the $y$-values of other vertices.
\end{enumerate}
\end{invariant}

Based on Invariant~\ref{inv:F-scale}, Edmonds' search will use the following Eligibility criterion:
\begin{criterion}\label{cri:scale}
At scale $i$, an edge $e \in E$ is eligible if one of the following holds
\begin{enumerate}[itemsep=0pt]
  \item $e \in E_B$ for some $B \in \Omega$.
  \item $e \not\in F$ and $yz(e) = w_i(e) - \delta_i$.
  \item $e \in F$ and $yz(e) - w_i(e)$ is a nonnegative integer multiple of $\delta_i$.
\end{enumerate}
\end{criterion}

This is similar to Criterion~\ref{cri:approx} except for we have a relaxed criterion for when $e \in F$. This relaxation is due to the fact that tightness is weakened at the end of each scale, and the eligibility criterion is then relaxed to accommodate it. We argue below that this relaxation does not affect the correctness of Edmonds' Search.

Before the start of scale $0$, the algorithm initializes $F, \Omega, y, z$ similar to the algorithm for small edge weights: $y(u) \leftarrow W/2$, $\Omega \leftarrow \emptyset$, $F \leftarrow \emptyset$. At scale $i$, the duals of unsaturated vertices start at $W/2^{i+1}$. We execute $(W/2^{i+2}) / (\delta_i/2) = O(\epsilon^{-1})$ iterations of Edmonds' search with parameter $\delta_i$, using Criterion~\ref{cri:scale} of eligibility. The scale terminates when the $y$-values of unsaturated vertices are reduced to $W/2^{i+2}$, or, in the last iteration, as they reach $0$.

Notice that although the invariant and the eligibility criterion are changed, the fact that Edmonds' search preserves the complementary slackness invariant still holds. The proof of Lemma~\ref{lem:dual-adjustment} goes through, as long as the definition of eligibility guarantees the following parity property:

\begin{lemma} \label{lem:gap}
At any point of scale $i$, let $\overline{S}$ be the set of vertices in $\Gelig$ reachable from an unsaturated vertex using eligible edges. The $y$-value of any vertex $v \in V$ with $B_v \in \overline{S}$ has the same parity as a multiple of $\delta_i/2$.
\end{lemma}

We omit the proof of Lemma~\ref{lem:gap}. The details are similar to Lemma~\ref{lem:dual-adjustment}, using
Criterion~\ref{cri:scale} in lieu of Criterion~\ref{cri:approx}.

Now we sketch why Criterion~\ref{cri:scale} ensures Invariant~\ref{inv:F-scale}, in particular, how it ensures approximate domination and approximate tightness. We will not prove it formally as the details are very similar to Lemma~\ref{lem:dual-adjustment} and Lemma~\ref{lem:ag-bf-bd}.

Observe that primal and dual variables initially satisfy Invariant~\ref{inv:F-scale}, in particular parts
$2$ and $3$. This is because all edges have $yz$-values equal to $W$, and no edge is in $M \bigcup_{B \in \Omega}$ $E_B$.

Notice that dual adjustment never changes the $yz$-values of edges inside any blossom $B \in \Omega$, while it will have the following effect on edge $e$ if its endpoints
lie in different blossoms.
\begin{enumerate}[itemsep=0ex]
  \item If $e \not\in F$ and is ineligible, $yz(e)$ might decrease but will never drop below the threshold for eligibility, i.e., it will not drop below $w_i(e) - \delta_i$.
  \item If $e \not\in F$ and is eligible, $yz(e)$ will never decrease.
  \item If $e \in F$ and is ineligible, $yz(e)$ might increase but will never exceed the threshold for eligibility, i.e., it will not raise above $w_i(e) + 2\delta_{j_e} - 2\delta_i$.
  \item If $e \in F$ and is eligible, $yz(e)$ will never increase.
\end{enumerate}

In other words, with the proper definition of Eligibility, Dual Adjustment will not destroy approximate domination and approximate tightness. Therefore Edmonds' search within scale $i$ will preserve Invariant~\ref{inv:F-scale}.

We also need to manipulate the duals between different scales to ensure Invariant~\ref{inv:F-scale}. Formally, after completion of scale $i$, we  increment all the $y$-values by $\delta_{i+1}$, i.e.,
if $yz'$ and $yz$ are the function before and after
dual adjustment, $yz(e)=yz'(e)+2\delta_{i+1}$.
No change is made to $F, \Omega$ and $z$. This will ensure both approximate domination and approximate tightness hold at scale $i+1$.
At the previous scale we have approximate domination $yz(e) \geq w_i(e) - \delta_i$.
The weights at scale $i$ and $i+1$ satisfy
$w_{i+1}(e) \leq w_i(e) + \delta_{i+1}$.
Thus, after dual adjustment,
\begin{align*}
    yz(e) &= yz'(e) + 2\delta_{i+1} \\
          &\geq w_i(e) - \delta_i + 2\delta_{i+1} \\
          &\geq w_{i+1}(e) - \delta_{i+1} - \delta_i + 2\delta_{i+1} \\
          &= w_{i+1}(e) - \delta_{i+1}.
\end{align*}
For approximate tightness, we have
\[
yz(e) - w_{i+1}(e)
\leq yz(e) - w_i(e)
= yz'(e) - w_i(e) - 2\delta_{i+1}
\leq 2\delta_{j_e} - 2\delta_i + 2\delta_{i+1}
= 2\delta_{j_e} - 2\delta_{i+1},
\]
since
$\delta_{i+1} = \delta_i/2$.
This step is the reason for the definition of 
Invariant~\ref{inv:F-scale}(3), as approximate tightness is gradually relaxed in this step.
The algorithm terminates when the $y$-values of all unsaturated vertices reach $0$. It terminates with an $f$-matching $F$ and its corresponding duals $y$, $z$ and $\Omega$ satisfying the following property:
\begin{property}[Final Complementary Slackness]~\label{pro:F-scale}
  \begin{enumerate}[itemsep=0pt]
      \item {\it Approximate Domination.} For all $e \not\in F$ or $e \in E_B$ for any $B \in \Omega$, $yz(e) \geq w(e) - \delta_L$.
      \item {\it Approximate Tightness.} For all $e \in F\cup (\bigcup_{B\in\Omega} E_B)$, let $j_e$ be the index of the last scale that $e$ joined
      $F\cup (\bigcup_{B\in\Omega} E_B)$.
      We have $yz(e) \leq w(e) + 2 \delta_{j_e}$.
      \item {\it Blossom Maturity.} For all blossoms $B \in \Omega$, $|F \cap (\gamma(B) \cup I(B))| = \left\lfloor \frac{f(V)+ |I(B)|}{2} \right\rfloor)$.
      \item {\it Unsaturated Vertices' Duals.} The $y$-values of all unsaturated vertices are $0$.
  \end{enumerate}
\end{property}

This implies approximate domination and approximate tightness are satisfied within some factor $1\pm O(\epsilon)$. For approximate domination this is easy to see since $w(e) \geq 1$ and $\delta_L = \epsilon/2$, thus $yz(e) \geq (1 - \epsilon/2) w(e)$ if $e \not\in F$.
For approximate tightness, we can lower bound the weight of $e$ if $e$ last entered $F$
or a blossom at scale $j=j_e$.
Throughout scale $j$, the $y$-values are at least $W/2^{j+2}$, so $w(e) \ge w_j(e) \ge 2(W/2^{j+2}) - \delta_j$.
Since $\delta_j = \epsilon W/2^j$,  $yz(e) \le w(e) + 2\delta_j \leq (1 +4\epsilon) w(e)$ when $e \in F$.

To sum up, the $f$-matching satisfies Property~\ref{pro:F-mult-approx} 
with parameters $(\epsilon/2,4\epsilon)$, and Lemma~\ref{lem:reduction} 
implies the complementary $f'$-edge cover satisfies Property~\ref{pro:EC-mult-approx} with parameters $(4\epsilon,\epsilon/2)$.
Lemmas~\ref{lem:F-mult-approx} and \ref{lem:ec-mult-cs} imply 
the approximation factors for the $f$-matching
and $f'$-edge cover are, respectively, 
$(1-\epsilon/2)(1+4\epsilon)^{-1}$ and
$(1+4\epsilon)(1-\epsilon/2)^{-1}$.

The running time of the algorithm is $O(m\epsilon^{-1} \log W)$ because there are $\log W + 1$ scales, and each scale consists of 
$O(\epsilon^{-1})$ iterations of Edmonds' search, 
which can be implemented in $O(m\alpha(m, n))$ time.

\begin{theorem} ~\label{thm:general-weight}
A $(1 - \epsilon)$-approximate maximum weight $f$-matching can be computed in time 
$O(m\alpha(m, n)\epsilon^{-1}\log W)$.
\end{theorem}

\begin{corollary} ~\label{cor:general-weight}
A $(1 + \epsilon)$-approximate minimum weight $f$-edge cover can be computed in time
$O(m\alpha(m, n)\epsilon^{-1}\log W)$.
\end{corollary}

\subsection{An $O_\epsilon(m\alpha(m, n))$ Algorithm}
We also point out that by applying techniques in \cite[\S 3.2]{DP14}, the algorithm can be modified to run in time independent of $W$.
The main idea is to force the algorithm to ignore an edge $e$ for all but $O(\log \epsilon^{-1})$ scales. First, we index edges by the first scale that it can ever become eligible. Since at scale $i$, $y$-values can drop at most to $W / 2^{i+1}$, any edge with weight below $W / 2^i$ cannot be eligible at scale $i$. Let $\mu_i = W/2^i$ and $\scale(e)$ be the unique $i$ such that
$w(e) \in [\mu_i, \mu_{i-1})$.  Notice that we can ignore $e$ in any scale $j < \scale(e)$. Moreover, we will also forcibly ignore $e$ at scale $j > \scale(e) + \lambda$ where $\lambda = \log \epsilon^{-1} + O(1)$. Ignoring an otherwise eligible edge might cause violations of approximate tightness and approximate domination. However, since the $y$-values of free vertices are $O(\epsilon w(e))$ at this point, this violation
will only amount to $O(\epsilon w(e))$.

To see this, notice that $\mu_i$ is also an upper bound to the amount of change to $yz(e)$ caused by all Dual Adjustment \emph{after} scale $i$. Hence, after scale $\scale(e) + \lambda$, the total amount of violation to approximate tightness and approximate domination on $e$ can be bounded by $\mu_{\scale(e)+\lambda} = O(\epsilon) \mu_{\scale(e)} = O(\epsilon) w(e)$, which guarantees we still get a $(1-O(\epsilon))$-approximate solution.

Therefore, every edge takes part in at most $\log \epsilon^{-1} +O(1)$ scales, with $O(\epsilon^{-1})$ cost per scale. The total running time is $O(m \alpha(m, n)\epsilon^{-1} \log \epsilon^{-1})$.

\begin{theorem} \label{thm:linear}
A $(1 - \epsilon)$-approximate maximum weight $f$-matching and a $(1 + \epsilon)$-approximate minimum weight $f$-edge cover can be computed in $O(m\alpha(m, n)\epsilon^{-1}\log \epsilon^{-1})$ time, independent of the weight function.
\end{theorem}

\section{An $O(m\alpha(m, n))$ Time 
Augmenting Walk Algorithm} \label{sec:aug-walk}

In this section, we show how to implement the augmentation and blossom formation steps in $O(m \alpha(m, n))$ time. 
The goal of the augmentation step is to find a set of augmenting walks \underline{\emph{and}} alternating 
cycles in the contracted eligible subgraph, such that after the removal of these cycles and walks, the subgraph no longer contains any augmenting walks. In the blossom formation step, we are given a contracted graph without any augmenting walks. The goal is to find a maximal set of \emph{vertex-disjoint} reachable and contractable blossoms, i.e., a set of blossoms whose contraction will leave the graph without any reachable and contractable blossoms.

We formalize these as the following two problems:

\begin{definition} \label{def:p1}
In the \emph{Maximal Disjoint Walks and Cycles Problem} and the \emph{Maximal Disjoint Blossoms Problem}, we are given a graph $G = (V, E)$, where $V$ is partitioned into two sets $V_s$, $V_b$,
an $f$-matching $M$,
and a partial function
$\eta: V_b \mapsto E$ such that
$\eta(v) \in \delta(v)$ if $\eta(v)$ exists.
Here $v\in V_b$ represents a contracted blossom
and $\eta(v)$ the incident base edge, if any.
The goal of the \emph{Maximal Disjoint Walks and Cycles Problem} is to find a set of augmenting walks 
$\Psi$ and a set of alternating cycles $\mathcal{C}$, 
where all cycles and walks are mutually disjoint, such that after removing all edges in $\mathcal{C}$ and $\Psi$, 
the remaining graph $G$ does not contain any augmenting 
walks. 
For the \emph{Maximal Disjoint Blossoms}, we are guaranteed that the graph does not contain any augmenting walks. 
The goal is to find a maximal set $\Omega$ of \emph{vertex disjoint} reachable blossoms. 
\end{definition}

There are several subtleties in this definition. $G$ here is used as a contracted graph obtained by contracting a set of nested blossoms $\Omega$ from an underlying graph. Therefore, augmenting walks and alternating cycles are defined according to Definition~\ref{def:AP-f-factor}
and the definition of \emph{alternation} from Section~\ref{sec:small-weight}, by treating $\eta(v)$ as $v$'s base edge when $v$ represents a nontrivial blossom.
As a result, an alternating cycle in $G$ might not have even length in $G$ and an augmenting walk might not have odd length in $G$. It is guaranteed, by Lemma~\ref{lem:blossom-walk}, that the preimages of these cycles and walks 
in the underlying graph are odd and even, respectively.

Moreover, it is \emph{not} guaranteed that no augmenting walk exists in the underlying graph after removing the image of $\Psi$ and $\mathcal{C}$ in it.\footnote{This is because multiple augmenting walks in the underlying graph can intersect a single blossom in $\Omega$ before we contract the blossom, while after contracting a blossom, any augmenting walk or alternating cycle going through the blossom will forbid the other walks and cycle to use the same blossom again (as it must go through the base edge).} However, it is sufficient since in the proof of Lemma~\ref{lem:dual-adjustment}, we only use the fact that the \emph{contracted graph} does not contain any augmenting walks.

The maximal disjoint paths and cycles problem is noticeably different from the problem solved in \cite[\S 8]{GT91} for $1$-matching. Instead of looking for a \emph{maximal} set of \emph{vertex disjoint} augmenting paths, we look for a set of \emph{edge disjoint} augmenting walks $\Psi$ in conjunction with a set of \emph{alternating cycles} $\mathcal{C}$ whose removal removes all augmenting walks from $G$.

Both algorithms try to search for a set of augmenting paths/walks by building an alternating structure $S$ (not necessarily a tree) whose topology is defined in Section~\ref{sec:small-weight}. However, in $1$-matching, the search structure branches only at outer singletons and blossoms, while in $f$-matching, it also branches at inner singletons. As a result, when the search process reaches a vertex $v$, it also assigns $v$ an inner/outer tag to remind the algorithm whether it is looking for an unmatched/non-$\eta$ edge, or a matched/$\eta$ edge to continue extending the structure. A vertex can obtain both inner and outer tags, but only one of them is the \emph{primary} one where the search procedure uses it to decide which edge to explore next. If a vertex has two tags, then the non-primary one must be \emph{exhausted}, meaning the algorithm has already finished exploring all eligible edges with respect to that tag.

A key difference between $1$-matching and $f$-matching is that augmenting walks can be non-simple, i.e., they may contain an alternating cycle as a subwalk. Consequently, when the search process reaches an outer (inner) singleton $u$, it can potentially find, through an unmatched (matched) edge an inner (outer) singleton $v$ that has already been visited before in the same search, and proceed to discover an augmenting walk.
This phenomenon is illustrated in Figure~\ref{fig:cycle}.
If the algorithm intends to discover $(v_0, v_1, v_2, v_3, v_4, v_1, v_5, v_6)$ as an augmenting walk, it will reach $v_1$ with inner tag twice; first from $v_0$, then from $v_4$. Notice that in ordinary matching, edge $(v_1, v_5)$ cannot exist and edge $(v_4, v_1)$ is ignored as it provides no useful information regarding whether $v_1$ is an inner/outer vertex.

One might propose to ignore and discard the edge $(v_4, v_1)$ and return the simple path $(v_0, v_1, v_5, v_6)$. However, edges like $(v_4, v_1)$ cannot simply be discarded from future searches as they might participate in other augmenting walks, say
$(v_{10}, v_9, v_4, v_1, v_8, v_7)$ that is edge disjoint from $(v_0, v_1, v_5, v_6)$.  To achieve a near linear running time, it is essential that edges
of this type only get scanned a bounded number of times.

Following the spirit of DFS, we wish to maintain that every search path is not self-intersecting, i.e., each vertex is visited at most once
with an \emph{inner} tag, and once with an \emph{outer} tag.
Whenever we discover an edge that leads to a self-intersection (e.g. $(v_4, v_1)$ in Figure~\ref{fig:cycle}), we augment along the alternating cycle introduced by this edge (e.g. $(v_1, v_2, v_3, v_4, v_1)$) and thereby remove every edge on the cycle from the future searches. We backtrack to $v_1$ and the search continues, say, to the edge $(v_1, v_5)$. 
This action has the same effect as allowing augmentation along the non-simple augmenting walk 
$(v_0, v_1, v_2, v_3, v_4, v_1, v_5, v_6)$, 
but conceptually avoids a self-intersecting search structure and thus makes the analysis simpler.

\begin{figure} \label{fig:cycle}
  \centering
  \includegraphics[scale=1.0]{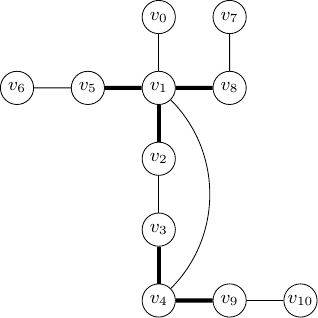}
  \caption{Example of a self-intersecting search structure and nonsimple augmenting walk. Here $v_0$ is the root of the search structure and $\{v_0, v_2, v_4, v_5\}$ is the set of outer vertices and $\{v_1, v_3\}$ is the set of inner vertices. The search begins with $v_0$ and proceed to $v_1, v_2, v_3, v_4$ in order. The procedure then scan the edge $(v_4, v_1)$ and because it connect an outer vertex to an inner vertex that is already visited, it might ignore the edge and backtrack to $v_1$ and return the augmenting walk $\langle v_0, v_1, v_5, v_6\rangle$. However, although $(v_1, v_4)$ is scanned and ignored, it cannot be discarded from future search as another augmenting walk, such as the dashed walk $\langle v_7, v_8, v_1, v_4, v_9, v_{10}\rangle$ might make use of the edge $(v_1, v_4)$ and $\Psi$ will not be maximal if $(v_1, v_4)$ participating in some augmenting walks.}
\end{figure}

\paragraph{Overview of the algorithm.}

We present one algorithm that solves both problems. The algorithm has two modes, the \emph{augmenting walk mode} and the \emph{blossom mode}. The augmenting walk mode receives a contracted graph as input and outputs a maximal set of augmenting walks and alternating cycles. The algorithm also discovers a set of nested blossoms in the process of searching for augmenting walks, but ignores these in the output. On the other hand, the blossom mode only takes in a contracted graph that does not have any augmenting walks. It uses the same algorithm to look for the set of augmenting walks, but ignores any alternating cycles it encounters. In outputs the set of nested blossoms discovered in the process. In this section, we will first present the algorithm in the augmenting walk mode in detail, then note 
the changes needed to 
put it in blossom mode.

\medskip

In this algorithm, we follow a standard recursive framework for computing a maximal set of edge disjoint paths, see \cite[\S 9]{Sch02}. The algorithm proceeds in phases. In phase $i$, $i \geq 1$, we choose a vertex $r$ that is still unsaturated after augmentations in phases $1, 2, \ldots, i-1$, and call a procedure \texttt{SEARCH-ONE} from this vertex.
\texttt{SEARCH-ONE} searches for an augmenting walk from $r$, and terminates with a possibly empty augmenting walk $P_i$ along with a set of disjoint alternating cycles $\mathcal{C}_i$. It guarantees that either $P_i$ is nonempty, or in the case when $P_i$ is empty, no augmenting walk starting from $r$ can reach another unsaturated vertex without intersecting $\mathcal{C}_i$. It then augments along $P_i$ and $\mathcal{C}_i$. The phase ends by discarding the set of edges encountered by the search procedure.

Formally, the input to \texttt{SEARCH-ONE} is a subgraph $G_i = (V_i, E_i)$ of $G$, an $f_i$-matching $M_i$ where $M_i \subseteq M$ and $f_i(v) \leq f(v)$ for all $v \in V_i$ , and an unsaturated vertex $r \in V_i$ with respect to $f_i$ and $M_i$. \texttt{SEARCH-ONE} finds an augmenting walk $P_i$ (possibly empty), a set of alternating cycles $\mathcal{C}_i$ and a set of edges $H_i \subseteq E_i$ that satisfy the following property.

\begin{property} \label{pro:max-ap}
Any augmenting walk that intersects $H_i$ must also intersect $P_i$ or a cycle in $\mathcal{C}_i$.
\end{property}

After \texttt{SEARCH-ONE} terminates, we finish the phase by removing $H_i$ from $G_i$. If $P_i$ is empty, we also remove the vertex $r$ from $G_i$. Let the resulting graph be $G_{i+1}$. Define the $f_{i+1}$-matching $M_{i+1}$ by $M_{i+1} = M_i \setminus H_i$ and
\[
f_{i+1}(v) =
\begin{cases}
f_i(v) - |M_i \cap H_i \cap \delta(v)| - 2|M_{i} \cap H_i \cap \delta_0(M_i)| \\
\hspace{1cm} \parbox[t]{0.8\textwidth}{If $v$ is not a terminal vertex of $P_i$.} \\
f_i(v) - |M_i \cap H_i \cap \delta(v)| - 2|M_{i} \cap H_i \cap \delta_0(M_i)| - 1 \\
\hspace{1cm} \parbox[t]{0.8\textwidth}{If $P_i$ is a nonempty nonclosed  augmenting walk and $v$ is one of the two terminal vertices of $P_i$.} \\
f_i(v) - |M_i \cap H_i \cap \delta(v)| - 2|M_{i} \cap H_i \cap \delta_0(M_i)| - 2 \\
\hspace{1cm} \parbox[t]{0.8\textwidth}{If $P_i$ is a nonempty closed  augmenting walk  that starts and ends with $v$.}
\end{cases}
\]

Conceptually, this change restricts the $f_i$-matching $M_i$ to the subgraph $G_{i+1}$ while maintaining the property that each vertex still has the same deficiency, except for the terminal vertices
of $P_i$, whose deficiencies are decremented by $1$
(or $2$ for closed walks) after augmenting along $P_i$.
Finally, the algorithm adds the path $P_i$ and cycles $\mathcal{C}_i$ to $\Psi$ and $\mathcal{C}$, respectively, and terminates phase $i$.

\paragraph{A detailed illustration of \texttt{SEARCH-ONE}.}

The call to \texttt{SEARCH-ONE} in phase $i$
maintains a laminar set of blossoms $\Omega_i$
over vertices in $G_i$.
Here $\Omega_i$ only contains the blossoms newly discovered in the search procedure and does not include the already contracted blossoms inherited from the input (vertices in $V_b$).
In this section, we use the word \emph{blossom} solely for the newly discovered blossoms in $\Omega_i$ and \emph{blossom vertices} for blossoms inherited from the input, i.e., vertices in $V_b$.
\emph{Singletons} still  refer to vertices in $V_s$.
We use $B(v)$ to denote the inclusion-wise maximal blossom in $\Omega_i$ that contains $v$ and $\beta(v)$ to denote the base of $B(v)$. If $v$ is not contained in any nontrivial blossom in $\Omega_i$, we define $B(v) = \{v\}$ and $\beta(v) = v$.
Each blossom also might have a base edge $\eta(B)$. 
We denote the search structure on $G_i$ with $S_i$, and use $T_i$ to denote the search structure obtained from $S_i$ by contracting all blossoms in $\Omega_i$. 
Similar to \cite[\S 8]{GT91}, the search structure $S_i$ is a subgraph of $G_i$ but not necessarily a tree, while we maintain that $T_i$ must be a tree. 
$T_i$ is represented by storing parent pointers.
If $B$ is a singleton or blossom in $\Omega_i$,
let $\tau(B)$ be the edge joining $B$ 
to its parent in $T_i$.

Blossoms are maintained using a data structure that supports the following operation: given a blossom $B$, a vertex $v$ in blossom $B$, and a bit $s \in \{0, 1\}$, the data structure returns the alternating walk $P_s(v)$ from $v$ to $\beta(B)$ whose existence is guaranteed by Lemma~\ref{lem:blossom-walk}, in time linear in the length of the walk. This can be done using simple bookkeeping as in Gabow's implementation for Edmonds algorithm~\cite{Gab76} and we leave the details to the readers.  We use a union-find data structure to find the outermost blossom containing a vertex, i.e., given $v\in V$,
find $B(v) \in \Omega_i$.  Furthermore, each blossom $B$ (outermost or not) maintains a pointer to its base $\beta(B)$.

\texttt{SEARCH-ONE} explores the graph in a depth first fashion:
The search begins at an unsaturated singleton or an unsaturated blossom vertex $r$ in $G_i$.
Similar to DFS, when the locus of the search is at $u$ we have an alternating walk $P(u)$ from $r$ to $u$ in $G_i$.
We call $u$ the \emph{active vertex} and $P(u)$ the \emph{active walk}. 
For efficiency purposes, we do not maintain the active walk explicitly.
Instead, we maintain a contracted active walk $\widehat{P}(u)$. 
The contracted active walk $\widehat{P}(u)$ is of the form $\langle B_0, e_0, B_1, e_1, \ldots, e_{k-1}, B_k\rangle$, where $r \in B_0$, $u \in B_k$ and each $B_j$ are either singletons or blossoms (not necessarily maximal) in $\Omega_i$. 
Each edge $e_j$ connects $B_{j}$ to $B_{j+1}$ and 
edges $e_j$ and $e_{j+1}$ \emph{alternate} at 
$B_{j+1}$ for all $0 \leq j < k-1$.
The active walk $P(u)$ can be reconstructed from $\widehat{P}(u)$ in time $O(|P(u)|)$ using the blossom data structure mentioned above. 
To maintain the property that the active walk is alternating, the algorithm also assigns inner and outer tags to each vertex in $S_i$ according to Definition~\ref{def:in-out}. 
There is one more issue that arises in our algorithm compared to the similar DFS routine in \cite[\S 8]{GT91}. Upon discovering and augmenting along an alternating cycle $C$, the search structure $S_i$ becomes \emph{disconnected} from the root since $C$ has been effectively removed from the graph. 
Therefore, we remove all vertices and edges that descend from $C$ in $T_i$ from $T_i$ and also remove their preimages from $S_i$.
However, we still maintain some useful information about these vertices and edges. This includes the blossom structure, the former parent pointers of vertices and blossoms, as well as all the tags that vertices and edges carry and whether they are exhausted or not.

Initially, the contracted active walk consists of a single vertex $r$, and $r$ is labelled outer.
At each iteration, the algorithm \emph{explores} a new edge $(u, v)$ incident to the active vertex $u$ that is \emph{eligible for} $u$ with respect to $u$'s current inner/outer tag. Notice that this immediately defines an inner/outer tag for $v$.
On exploring the edge $(u, v)$, the algorithm does one of the following depending on the location of $v$ with respect to the search structure, and the tag $v$ is carrying. 

\begin{enumerate}[itemsep=0pt]
  \item {\it Augmentation}: When $v$ is an unsaturated singleton and $(u, v)$ is unmatched, or $v$ is an unsaturated blossom vertex, or when $v = r$ is the root of the search tree and the deficiency of $v$ is at least $2$, $P(u) \circ (u, v)$ forms an augmenting walk. We extend the active path with $(u,v)$, terminate the search and set the active walk $P(v) = P(u) \circ (u, v)$ and $P_i = P(v)$. In this step, the edge $(u, v)$ is considered \emph{explored from $u$}.

  \item {\it DFS Extension}: If $v$ is a singleton and $v$ has never been exhausted before with the same tag, add $(u, v)$ to the search structure $S_i$ and make $B(v)$ a child of $B(u)$ in $T_i$. Set the active vertex to $v$ and extend the contracted active walk $\widehat{P}(v) = P(u) \circ (u, v) \circ B(v)$ accordingly. In this step, edge $(u, v)$ is considered \emph{explored from $u$}.
  
  \item {\it Cycle Cancellation}: If both $B(u)$ and $B(v)$ are in $T_i$ and $B(v)$ is an ancestor of $B(u)$ and $(u, v)$ is not eligible for $v$, the tree path from $B(v)$ to $B(u)$, along with the edge $(u, v)$, forms an alternating cycle $C$. We add $C$ to $\mathcal{C}_i$. Retract the active walk back to $v$.
  After this step, all edges $e \in C$ will be categorized as
  explored from \emph{both endpoints}.

  This step effectively disconnects all descendants of $C$ from the root.\footnote{It is still possible that later we discover some path from a descendant of $C$ back to the root that circumvents $C$} We remove this part from the search structure $S_i$ and $T_i$. However, old parent pointers that were used to maintain $T_i$ are still kept.
  
  \item {\it Blossom Formation}: If $B(u) \neq B(v)$,  $B(v)$ is not in $\widehat{P}(u)$, and edge $(u, v)$ is also eligible for $v$, then we can potentially form a new blossom. We do so by extending the active walk from $B(u)$ to $B(v)$, then to the chain of ancestors of $B(v)$ encountered when following parent pointers, until one of the following stopping conditions is met. 
  If we append $B(u)$ 
  to the active walk then we stop, with 
  a new blossom having been detected. The process may stop prematurely if 
  (i) the next parent edge is already marked \emph{explored} because it is in $\mathcal{C}_i$
  or (ii) if the head of the walk is exhausted according to its tag.
  
  Notice that this step puts new tags onto vertices in these blossoms and also attaches them onto the search tree $T_i$.
  Let $\tilde{P}$ be the suffix of $\widehat{P}$ starting from $B(v)$.
  If $\tilde{P}$ ends at $B(u)$, we have identified a blossom $B$ whose subblossoms consists of the set of blossoms in $\tilde{P}$. We contract the blossom by putting $B$ into $\Omega_i$ and update the union-find data structure accordingly, 
  then contract all vertices of 
  $\tilde{P}$ in $T_i$. 
  
  If $\tilde{P}$ ends prematurely before getting to $B(u)$, we cannot contract the blossom.  This search has effectively appended $\tilde{P}$ to $\widehat{P}$ with DFS extensions and the search proceeds from the last element of $\widehat{P}$ as usual.

  In both cases, edges are considered explored in the direction of the active walk when they first enter the active walk, and exhausted after they leave the active walk.

  \item {\it DFS Retraction}: If every edge $(u, v)$ eligible for $u$ has already been explored, retract from $u$ to its predecessor on the (contracted) active walk. If $u=r$ is the only vertex in the search path, terminate the search with $P_i = \emptyset$. Otherwise, let $w$ ($B(w)$) be the parent of $u$ in the (contracted) active walk. The edge $(w, u)$ is now considered \emph{exhausted from $w$}. This means that every edge eligible for $u$ is recursively exhausted and no augmenting walk can be found by following the active walk via the edge $(w, u)$. (It may still be possible to find an augmenting walk via edge $(w, u)$ when the search visits $u$ again in a blossom formation step and explore $(u, w)$ from $u$). Moreover, the vertex $u$'s primary tag is now consider \emph{exhausted}.

  \item {\it Null Exploration}: This step includes all scenarios where we explore the edge $(u, v)$ to no effect. This includes: when $B(v)$ is a descendant of $B(u)$ and $(u, v)$ is not eligible for $v$; when $B(v)$ is an ancestor of $B(u)$ and $(u, v)$ is eligible for $v$; or when $(u, v)$ is a cross edge in $T_i$. In these cases, we ignore the edge $(u, v)$ while still categorizing it as \emph{exhausted from $u$}.
\end{enumerate}

As stated above, each edge $(u,v)$ along with an endpoint of it,
say $u$, has one of three statuses at any point in the algorithm:
\begin{enumerate}
    \item \emph{Explored from $u$}: This means the search has visited the vertex $u$, extended the active walk from $u$ to $v$ using edge $(u, v)$, in either Augmentation, DFS extension, Blossom Formation, or Cycle Cancellation step.
    \item \emph{Exhausted from $u$}: This means the search has visited the vertex $u$, extended the search path to $v$ via $(u, v)$ and then backtracked to $u$ in DFS Retraction or Null Exploration steps.
    \item \emph{Unexplored from $u$}: If $(u, v)$ is not considered explored or exhausted from $u$, it is then unexplored from $u$. This means the search has either yet to visit $u$;
    or the search has visited $u$ but never extended the active walk using the edge $(u, v)$ because it is ineligible for $u$;
    or it is eligible but the search has yet to explore $(u, v)$.
\end{enumerate}

Finally, the edge set $H_i$ is the set of edges that are explored or exhausted from at least one of their endpoints during the search. Recall this is the set we remove from the graph $G_i$ before termination of a phase. This completes our description of the \texttt{SEARCH-ONE} procedure.
Now we state the set of invariants satisfied by 
\texttt{SEARCH-ONE}.

\begin{invariant} ~ \label{inv:basic}
\begin{enumerate}
    \setlength{\itemsep}{0pt}
    \item \emph{Structural Invariant of $S_i$:} $S_i$ consists of a subset of edges in $G_i$ that are visited during the search. Every vertex in $S_i$ carries an inner tag or outer tag or both. When a vertex is outside the active walk, all its tags are exhausted. When it is inside the active walk, one tag is the primary tag. The other tag, if it exists, must be exhausted. 
    Furthermore, if $v$ is inner, there exists an alternating walk from $r$ to $v$ that ends with an unmatched edge if $v \in V_s$ or a non-$\eta$ edge if $v \in V_b$. If $v$ is outer, the alternating walk terminates with a matched edge if $v \in V_s$ or an $\eta$ edge if $v \in V_b$. These walks avoid $\mathcal{C}_i$. 

    \item \emph{Structural Invariant of $T_i$:} $T_i$ is a contracted graph obtained from $S_i$ by contracting all inclusionwise-maximal blossoms in $\Omega_i$. $T_i$ must be a tree. Some blossoms in $\Omega_i$ might not be represented in $T_i$.

    \item \emph{Depth-first property of $S_i$:} The union of the active walk and the set of alternating cycles $\mathcal{C}_i$ consists of precisely the edges that are explored but not exhausted from at least one endpoint.
    If $(u, v)$ is an edge in $H_i$ but not in the active walk or the alternating cycles, then $(u, v)$ must be exhausted from $u$ or $v$. 
    Moreover, $S_i$ is disjoint from $\mathcal{C}_i$.

    \item \emph{Maximality} If $(u, v)$ is marked exhausted from $u$ while $(v, w)$ is an edge eligible for $v$, then the algorithm must have exhausted $(v, w)$ from $v$. This holds for all edges regardless whether they are in $S_i$ or not.
    
    \item \emph{Parent Pointers: } Fix any blossom $B$ in $\Omega_i$, possibly trivial and possibly not in $T_i$.  If the parent pointer $\tau(B)$ is defined, consider the path generated by following parent pointers from $B$, terminating when one of the following conditions is met: 
    (i) the path reaches $r$, 
    (ii) the next edge in the path would be in $\mathcal{C}_i$, 
    (iii) the last vertex in the path is exhausted w.r.t. the appropriate tag,
    or (iv) the last vertex in the path is in $\widehat{P}$.  This path is well defined
    and is alternating.
\end{enumerate}
\end{invariant}

\begin{lemma}
\emph{Augmentation}, \emph{DFS Extension}, \emph{Blossom Formation}, \emph{DFS Retraction}, \emph{Null Exploration} and \emph{Cycle Cancellation} all preserve Invariant~\ref{inv:basic}.
\end{lemma}

\begin{proof}
The first invariant follows from how we grow the search structure $S_i$ and active walk. When the active walk extends to a vertex $v$ with the current tag outer, the active walk must be an alternating walk ending with a matched edge or an $\eta$ edge. If the tag is inner, the active walk ends with a unmatched edge or a non-$\eta$ edge. This ensures that there exists an alternating walk from the root to each vertex in $S_i$ with a terminal edge that is consistent with its tag. Moreover, tags are labelled exhausted if and only if a vertex carrying the tag leaves the active walk by backtracking.

The second invariant follows from the fact that when we form a blossom in the Blossom Formation step, the constituent (subblossoms) in $\Omega_i$ always come from a connected ancestor-descendant path in $T_i$. Contracting a connected component in the tree will not create any cycle and thus $T_i$ remains a tree.

For the third invariant, observe that an edge becomes explored from an endpoint when it joins the active walk 
in a DFS Extension, Blossom Formation, or Augmentation step. It becomes exhausted when it leaves the active walk in a DFS Retraction, Null Exploration, or Cycle Cancellation step. Moreover, in the Blossom Formation step, since we are visiting vertices in descendant-to-ancestor direction, all edges in the active walk must remain explored and edges outside active walk are exhausted.
Therefore, any edge in $H_i$ that is not in the active walk or alternating cycles must be exhausted from at least one of its endpoints.
Lastly, in the Blossom Formation and Cycle Cancellation steps, we specifically enforce that the active walk never uses 
any edge inside $\mathcal{C}_i$.

For the fourth invariant, first notice that $(u, v)$ becomes exhausted via a DFS Retraction step or a Null Exploration step. In both cases the search must have retracted from $v$ to some vertices and therefore has explored and exhausted every edge eligible for $v$, including $(v, w)$. If $w$ is an unsaturated singleton and $(v, w)$ is unmatched, or $w$ is an unsaturated blossom vertex, an Augmentation step would have occurred when the algorithm explores $(v, w)$ and left the edge $(v, w)$ explored and not exhausted.

For the fifth invariant, consider a blossom $B$ and the path starting from $B$ following the parent pointers. We call this path the \emph{ancestral path} from $B$. By the inductive hypothesis, 
the ancestral path alternates until it 
ends by reaching $r$, or an exhausted vertex, 
or an alternating cycle edge, 
or the active walk.\footnote{Getting to the active walk does not automatically imply that you can then get to the root, since the ancestral path might not alternate at the vertex when it first reaches the active walk} 
Now consider how this path may change in a Cycle Cancellation, DFS Extension, DFS Retraction,
or Blossom Formation step.
In a Cycle Cancellation step, we might shorten the path by including one of its edges in an alternating cycle.
This preserves the invariant. In a DFS Extension or Blossom Formation step, the active walk might extend into the ancestral path, making the ancestral path terminate earlier.
This also preserves the invariant. 
A DFS retraction can remove the last vertex from $\widehat{P}$ that appears on the ancestral path.  That vertex is by definition exhausted for its tag type, so the ancestral path terminates at the same point as before, but for a different reason.  In all cases the invariant is preserved.  Finally, notice that for any blossom inside $T_i$, the ancestral path from that blossom always alternates until it reaches the root $r$. 
\end{proof}

Now we state the correctness of \texttt{SEARCH-ONE}:

\begin{lemma} \label{lem:search-one-cor}
When \texttt{SEARCH-ONE} terminates, if there is an augmenting path $P'$ that intersects $H_i$ at some edge $e$, then $P'$ must intersect $P_i$ or $\mathcal{C}_i$ at some edge.
\end{lemma}

\begin{proof}
Assume for contradiction that $P'$ is edge-disjoint with $P_i$ and $\mathcal{C}_i$. Let $P'$ intersect $H_i$ at some edge $(u_0, u_1)$. Since $(u_0, u_1)$ is not in $P_i$ or $\mathcal{C}_i$, $(u_0, u_1)$ must be exhausted in one of its directions, say from $u_0$. This makes $(u_0, u_1)$ eligible for $u_0$. Now let $(u_0, u_1, \ldots, u_k)$ be the subpath of $P'$ from $u_0$ to the terminal vertex $u_k$ of $P'$ in the $(u_0, u_1)$ direction. We use induction to show that for all $0 \leq i < k$, edges $(u_i, u_{i+1})$ must be eligible for $u_i$ and exhausted from $u_i$:

The base case $i = 0$ holds by our assumption. Now suppose $(u_i, u_{i+1})$ is exhausted from $u_i$ for some $i \geq 0$. Consider the edge $(u_{i+1}, u_{i+2})$.
It is necessary that $u_{i+1}$ was in the search structure $S$ when $(u_i, u_{i+1})$ is marked exhausted from $u_i$,
and at this moment $u_{i+1}$ is either inner or outer or both. In particular, at this moment $u_{i+1}$ must still own the tag that is consistent with the predecessor edge $(u_i, u_{i+1})$. 
Here, by consistent we mean the tag defined by Definition~\ref{def:in-out}, treating the edge $(u_i, u_{i+1})$ as $\tau(u_{i+1})$.
Combined with the fact that the edge $(u_i, u_{i+1})$ alternates with the edge $(u_{i+1}, u_{i+2})$, it is necessary that $(u_{i+1}, u_{i+2})$ be eligible for $u_{i+1}$.
Hence by Invariant~\ref{inv:basic}, it must also be exhausted from $u_{i+1}$.

This means the edge $(u_{k-1}, u_{k})$ must be eligible for $u_{k-1}$ and exhausted from $u_{k-1}$. Notice that the vertex $u_k$ and edge $(u_{k-1}, u_k)$ must satisfy the terminal vertex and edge requirement in Definition~\ref{def:AP-f-factor}. But in this case, an augmenting walk would have been formed when the algorithm was exploring the edge $(u_{k-1}, u_k)$ from $u_{k-1}$, which put the edge in $P_i$, which is a contradiction.
\end{proof}

This gives the following lemma.
\begin{lemma} \label{lem:search-one}
\texttt{SEARCH-ONE} finds in time $O(m_i \alpha(m_i, n_i))$ an edge set $H_i$, a set of alternating cycles $\mathcal{C}_i$ and an augmenting walk $P_i$ such that any augmenting walk $P'$ disjoint from $\mathcal{C}_i$ that intersects $H_i$ must also intersect $P_i$. Here $m_i$ and $n_i$ are 
the number of edges and vertices in $H_i$.
\end{lemma}
\begin{proof}
The correctness of \texttt{SEARCH-ONE} is argued in Lemma~\ref{lem:search-one-cor}. For running time, notice that each edge we examined is always classified as explored or exhausted from at least one of its endpoints.
Thus, the total number of edge examinations is $O(m_i)$.
The only non-trivial data structure needed is one for maintaining
the set of blossoms, which takes 
$O(m_i\alpha(m_i,n_i))$ time 
with the standard union-find algorithm~\cite{Tarjan75}. 
For reconstructing the active walk, we can use the bookkeeping labelling in \cite[\S 8]{GT91}, which enable reconstruction of the augmenting walk in time linear in the length of the walk.
\end{proof}

\begin{lemma} \label{lem:lin-ap}
We can find in $O(m\alpha(m, n))$ time 
a set of augmenting walks $\Psi$ and a set of 
alternating cycles $\mathcal{C}$ such that any 
augmenting walk $P'$ must intersect $\Psi$ 
or $\mathcal{C}$.
\end{lemma}
\begin{proof}
This algorithm can be seen as a recursive algorithm that first calls \texttt{SEARCH-ONE} on an input graph $G_1 = G$, finding an edge set $H_1$, a set of alternating cycles $C_1$ and an augmenting walk $P_1$. It removes $H_1$ from $G_1$ and the corresponding part in the $f$-matching to obtain $G_2$. Then it recurses on $G_2$. Let $\mathcal{C}'$ and $\Psi'$ be the output of the recursive call. We output $\mathcal{C} = \mathcal{C}_1 \cup \mathcal{C}'$ and $\Psi = \Psi' \cup \{P_1\}$.

By induction, any augmenting walk $P'$ in $G_2$ must intersect $\Psi'$ or $\mathcal{C}'$. Now suppose the augmenting walk $P'$ contains an edge in $H_1$. By Lemma~\ref{lem:search-one}, $P'$ must intersect $P_1$ or $\mathcal{C}_1$. Therefore $P'$ must intersect $\Psi$ or $\mathcal{C}$.
\end{proof}

\paragraph{Finding a maximal set of nested blossoms}

Now we present the blossom mode of the algorithm. The blossom mode is the same as the augmenting walk mode except for one step: It ignores alternating cycles it encounters. More specifically, when encountering an edge $(u, v)$ from $u$ in the search tree where $B(v)$ is the ancestor of $B(u)$, and $(u, v)$ is eligible for $u$ but not for $v$, we will not invoke the cycle cancellation step and add the fundamental cycle associated with $(u, v)$ to the set $\mathcal{C}$. Instead, we label the edge $(u, v)$ as exhausted from $(u)$, which removes it from any future search as it does not provide any useful information for the blossom search. 

Notice that since the input graph does not contain any augmenting walks, the status of each edge in the graph upon termination of a \texttt{SEARCH-ONE} execution is either unexplored or exhausted. By construction, the search trees rooted at each unsaturated vertex are necessarily edge disjoint. Moreover, since a vertex cannot be visited twice in two different search tree with different labels or we will have an augmenting walk, when a vertex is visited a second time from another search tree, all eligible edges must already be exhausted and the search will always retract immediately. Therefore, we can simply remove any visited vertex in a search tree from any future search and have these search trees be vertex disjoint.

Now we argue that the blossom set $\Omega'$ contains a maximal set of reachable blossoms. 
Consider a blossom $B$ that is reachable from one of the unsaturated vertices after the augmentation step. Using a similar induction to the one in Lemma~\ref{lem:search-one-cor}, we can show that all edges in the path from the root to the blossom $B$, and all edges in $B$, 
must be explored by the algorithm. Moreover, by the structure of the blossom, one of the edges inside the blossom that are incident to the base must be eligible for both endpoints at some point of the search and therefore must be explored in both directions. We show that the two endpoints of this edge must already be in a blossom in $\Omega'$.

\begin{lemma} ~\label{inv:same-blossom}
If $(u, v)$ is an edge that is once explored from both $u$ and $v$, then $u$ and $v$ must be in the same blossom in $\Omega'$.
\end{lemma}

\begin{proof}
Notice that in depth-first search, when the algorithm explores $(u, v)$ from both directions, $B(u)$ and $B(v)$ must be in an ancestor-descendant relationship in $T_i$. 
Now without loss of generality, we assume $(u, v)$ is first explored from $u$. If $v$ is a descendant of $u$, since the search has already backtracked from $v$, the only way that $v$ enters the search again is by a blossom formation step from ancestor of $u$ to a descendant of $v$, making them in the same blossom. If $v$ is an ancestor of $u$, when $(u, v)$ is explored from $v$, i.e., when the search backtracks from $u$ to $v$, $u$ must still be a descendant of $v$ because any blossom step in this process will not change the ancestor-descendant relation between $u$ and $v$. Then we have a blossom step triggered by $(u, v)$ and put them in the same blossom.
\end{proof}

This necessarily show that after contracting blossoms $\Omega'$, there are no more reachable blossom in the graph, so $\Omega'$ is a maximal set of blossoms. 
\section{Conclusion}

We present $O_\epsilon(m\alpha(m, n))$-time algorithms
for $(1-\epsilon)$-maximum weight $f$-matching and
$(1+\epsilon)$-approximate minimum weight $f$-edge cover. 
The main contributions of this work are approximation-preserving reductions between
$f$-matching and $f$-edge cover, and 
a version of Edmonds Search for generalized matching 
under relaxed complementary slackness conditions.

Using similar argument to the one found in~\cite{GT91}, 
it is also possible to show that this algorithm will lead to an $O(\sqrt{f(v)}m\alpha(n,m))$ time algorithm for \emph{maximum cardinality $f$-matching} and \emph{minimum cardinality $f$-edge cover}. The idea is to first treat 
the cardinality problem as a unit-weighted
matching problem. By setting $\epsilon = 1/\sqrt{f(V)}$, we can obtain an $f$-matching (edge cover) with total surplus (deficiency) $O(\sqrt{f(V)})$. Then we discard the duals 
and blossoms and use any linear time 
algorithm for finding one augmenting walk (reducing walk) at a time to complete the matching (edge cover). The total running time is $O(\sqrt{f(V)}m\alpha(m,n))$.
\newpage

\bibliographystyle{abbrv}
\bibliography{ref}

 \newpage

\end{document}